\documentclass[twocolumn]{svjour3}          
\smartqed  

\usepackage{setspace}
\usepackage{caption}
\usepackage{amsmath}
\usepackage[Algorithm]{algorithm}
\usepackage{algpseudocode} 
\usepackage{epstopdf}
\usepackage{color}
\usepackage{amssymb}
\usepackage[figuresright]{rotating}

\begin{document}

\title{A Two-Phase Scheme for Distributed TDMA Scheduling in WSNs with Flexibility to Trade-off  between Schedule Length and Scheduling Time}

\titlerunning{Two-Phase TDMA-Scheduling}        

\author{Ashutosh Bhatia       \and
        R. C. Hansdah 
}


\institute{Ashutosh Bhatiar \at
               Dept. of CSIS, BITS Pilani \\
               \email{ashutosh.bhatia@pilani.bits-pilani-ac.in} 
               \and
               R. C. Hansdah \at
              Dept. of CSA, IISc Bangalore \\
	   \email{hansdah@csa.iisc.ernet.in}
}

\date{}

\maketitle

\begin{abstract}
The existing distributed TDMA-scheduling techniques can be classified as either static or 
dynamic. The primary purpose of static  TDMA-scheduling algorithms is to improve the
channel utilization by generating a schedule of shorter  length. But, they usually take longer time 
 to schedule, and hence, are not suitable for WSNs, in which the network topology changes 
dynamically. On the other hand, dynamic TDMA-scheduling algorithms generate a schedule quickly, 
but they are not efficient in terms of generated schedule length. In this paper, we propose a new approach to TDMA scheduling for WSNs, that bridges the gap between the above two extreme types of 
TDMA-scheduling techniques, by  providing the flexibility to trade-off between the schedule length and the 
time required to generate the schedule (scheduling time). The proposed TDMA scheduling works in two phases. In the first phase, 
we generate a TDMA schedule quickly, 
which need not have to be  very efficient in terms of schedule length. In the second phase, we
iteratively reduce the schedule length in a manner, such that the process of schedule length 
reduction can be terminated after the execution  of an arbitrary number of iterations,
and still be left with a valid schedule. This step provides the capability to trade-off between schedule length and  scheduling time.
We have used  Castalia network simulator to evaluate the performance of proposed TDMA-scheduling scheme.
The simulation result together with theoretical analysis shows that, 
        in addition to the advantage of trading-off  the  schedule length with scheduling time, 
        the proposed TDMA scheduling approach achieves better performance than existing 
        algorithms in terms of  schedule length and scheduling time. 
\keywords{Wireless Sensor Networks \and  Media Access Control \and TDMA slot scheduling}
\end{abstract}

\section{Introduction}
The collision of frames severely degrades the performance of WSNs in terms of delay, channel utilization and power saving 
requirement. 
Time Division Media Access (TDMA) is a well known technique
to provide collision-free and energy-efficient transmission,
especially for applications with predictable communication
patterns. 
Furthermore, TDMA-based communication provides guaranteed QoS in terms 
of delay on the completion time of data collection, for instance, in the  timely 
detection of events in WSNs. 

In TDMA-based channel access,  time is divided into slots, and the slots are 
further organized into frames.
The slot(s) at which a node can transmit 
is usually determined by a predefined TDMA schedule.
The problem of finding a TDMA schedule with  minimum
schedule length (optimal TDMA schedule) is NP-Complete \cite{TDMA-WSN}. Plenty
of research work has been carried out to provide efficient
algorithms to perform TDMA scheduling. These algorithms
can be classified either as centralized or decentralized (distributed). The centralized approach normally needs complete
topology information at a single node in the WSN to perform
scheduling, and therefore it is not feasible for large-scale
multi-hop WSNs.

The existing distributed TDMA-scheduling techniques can be classified as either static or 
dynamic. The primary purpose of static  TDMA-scheduling algorithms is to improve the
channel utilization and possibly delay,  by generating a schedule of shorter length. 
Usually, static TDMA-scheduling algorithms take a very long time to 
generate such a schedule, and therefore, 
they are suitable for the situations in which the same schedule can be used for a 
sufficiently longer duration of time. 
But, sometimes, the same schedule cannot be reused for multiple
future data sessions, because  the network topology
may get changed with the progress of time, due to dynamic channel conditions or periodic  sleep scheduling
of sensor nodes to conserve their energy.
Additionally, even if the network topology has not changed, using the 
same schedule would not be efficient,  when the underlying application or the next-hop information used for forwarding 
(routing) the data changes. In such cases,  re-scheduling has to be performed after a certain period of time, and 
therefore, taking very long time to generate a compact schedule may lead to the consumption of more energy and result in increased 
delay, instead of improving the same.

As opposed  to static TDMA-scheduling algorithms, the algorithms which belong to dynamic category 
try to generate a TDMA schedule quickly, which may 
not be very efficient in terms of schedule length. Although they perform poorly in the terms of 
bandwidth utilization, they are suitable for the
cases where re-scheduling needs to be performed frequently due to the reasons discussed above.

The discussion given above  suggests that, an efficient TDMA-scheduling algorithm should try to minimize 
schedule length and scheduling time simultaneously. Unfortunately, both these objectives are 
mutually conflicting. Therefore, the trade-off between these two conflicting objectives needs to be 
addressed to improve the bandwidth utilization, energy saving  and QoS performance of a 
TDMA-scheduling algorithm. Although there exists some works that address the  joint objectives 
such as trade-off between energy efficiency and latency, but none of them have tried before the 
trade-off between schedule length and scheduling time.

The existing static algorithms for TDMA scheduling  are not designed in a manner so that
their execution can be stopped in-between to restrict the scheduling time, and still get a valid TDMA schedule. 
We call such algorithms as  single phase TDMA-scheduling algorithms in the sense that they can 
produce a valid TDMA schedule only at the end of the execution.
On the other hand, the existing dynamic TDMA-scheduling algorithms are not designed in a manner such that, 
multiple back-to-back execution of the same algorithm can be used to improve its performance in terms of 
schedule length of generated schedule. In summary, both types of TDMA-scheduling algorithms (static and dynamic), available in literature, 
do not have the flexibility to trade-off schedule length with scheduling time.

In this paper, we present a novel two-phase scheme for distributed TDMA scheduling in WSNs, 
that  bridges the gap between these two extreme (static and dynamic) types of TDMA-scheduling algorithms  by providing the flexibility to trade-off  schedule length with  scheduling-time, and also,
provides a better performance in terms of  schedule length and scheduling time than the existing algorithms.

 The rest of the paper is organized as follows. Section 2 discusses the related work. In section 3, we give the overview of 
 proposed two-phase scheme for TDMA Scheduling. In section 4, we present a distributed 
 and randomized TDMA-scheduling (RD-TDMA) algorithm as part of  phase 1 of the proposed two-phase TDMA-scheduling scheme. 
Section 5 presents a distributed schedule length  reduction (DSLR) algorithm for phase 2 of the proposed scheme.  The proof of correctness of RD-TDMA and DSLR algorithms are given in section 6. The 
 theoretical analysis, for the time taken by the RD-TDMA and DSLR algorithm, is provided in section 7. Results of simulation
studies of proposed two-phase scheme for distributed TDMA scheduling in WSNs and  its performance comparison with existing  works are discussed in section 8. 
Section 9 concludes the paper with suggestions for future work.

\section{Related Work} 
Previous works  \cite{TDMA-WSN,Centralized-1,Centralized-3,Centralized-4} on TDMA slot 
scheduling  primarily focus on decreasing 
the length of schedules. They are centralized in nature, and therefore, are not scalable. 
In \cite{TDMA-WSN},  specific scheduling problem for wireless sensor network, viz. 
converge-cast transmission,  is considered, where the scheduling
problem is to find a minimum length frame during which all nodes can send their packets to the
access point (AP), and the problem  is shown to be NP-complete.
The cluster based TDMA protocols (e.g., the protocols in \cite{LEACH,cluster}), prove to be 
having good scalability. 
The common feature of these
protocols is to partition the network into some number of clusters, in which each  cluster head is responsible 
for scheduling its
members. 
However, they suffer from inter cluster transmission interference because clusters 
created by distributed clustering 
algorithms are often overlapped, and several cluster heads may cover the same nodes. 
The protocol in \cite{contention-free-MAC} proposes a contention-free MAC for 
\textit{correlated-contention} which does not 
assume global time reference.
The approach  is based on local framing, where each node selects 
a slot in its own frame such that  slot  of any 2-hop-neighbor nodes must not overlap the selected slot. 
The protocol assumes that 
a node can 
detect a collision if two or more nodes within its transmission range  attempt to transmit at the same time. This 
approach has its own practical
limitations with wireless transceivers.	
A randomized CSMA protocol, called ``Sift'' \cite{sift}, tries to reduce the latency for delivering 
event reports 
instead of completely 
avoiding the collision. 
Sift uses a small and fixed contention window of size 32 slots and geometrically increases
non-uniform 
probability distribution for picking a transmission slot in the contention window. 
The key difference with the 
traditional MAC protocol, 
for example 802.11 \cite{ref-80211}, is that the probability distribution for selecting a contention slot is not 
uniform.
Moscibroda et al. \cite{Distance-1} have proposed a distributed graph coloring scheme
with a time complexity of $O(\rho \log n)$, where $\rho$ is the maximum node degree and n is the 
number of nodes in the 
network. The scheme performs
distance-1 coloring such that the adjacent nodes have different
colors. Note that,  this scheme does not prevent nodes within two
hops of each other from being assigned the same color potentially causing hidden terminal 
collisions between such nodes.
The NAMA \cite{NAMA} protocol uses a distributed scheduling scheme based on  hash function to 
determine the
priority among 
contending neighbors.  A major limitation of this hashing
based technique is that even though a node gets a higher priority in one neighborhood, it may still 
have
a lower priority in other neighborhoods. Thus, the maximum  number of slots could be of order $O(n)$, 
where $n$ is the number 
of nodes in the network. Secondly, since each node 
calculates the 
priority of all its two-hop neighbors 
for every slot, it leads to $O(\delta^2)$ computational complexity, where $\delta$ is the maximum 
size of two-hop neighborhood,
and hence, the scheme is not scalable for large network with resource constraint nodes. 
Herman et al. have  proposed a distributed TDMA slot assignment algorithm in \cite{distance-2} based
on  distance-2 coloring scheme.  In this algorithm, each node
maintains the  state information within its three-hop neighborhood, which
could be quite difficult and resource intensive. 

The  distributed TDMA slot scheduling algorithm, called DRAND  \cite{DRAND}, uses a 
distributed and randomized time
slot scheduling scheme which is used within a
MAC protocol, called Zebra-MAC \cite{Z-MAC}, to improve performance in sensor networks by combining 
the strength of
scheduled access during high loads and random access
during low loads. The Runtime complexity of DRAND is of the order of $O(\delta^2)$ due to unbounded 
message delays .  
The protocol in \cite{CCH} presents a distributed  slot assignment algorithm, which uses a 
heuristic 
approach, called 
Color Constraint Heuristic (CCH), for choosing the order in which  to color the nodes in a graph. This is 
unlike 
the DRAND algorithm which does not impose any ordering on the nodes to color them. Although the 
CCH scheme presented in \cite{CCH}
takes lesser number of slots as compared to DRAND,  the time taken by this scheme  to schedule 
all the nodes in the network is 
larger than that of DRAND.

The Five Phase Reservation Protocol (FPRP) in \cite{FPRP} is a distributed heuristic 
TDMA slot assignment algorithm. FPRP protocol is designed for dynamic slot assignment, in which 
the real time is divided into a 
series of a pair
of reservation and data transmission phases. For each time slot of the data transmission
phase, FPRP runs a five-phase protocol for a number cycles to pick a winner for each 
slot.

In another distributed slot scheduling algorithm, called DD-TDMA \cite{DTDMA}, a node $i$ decides 
slot $j$ as its 
own slot if all the nodes with id less than the
id of node $i$ have already decided 
their slot, where $j$ is the minimum available slot. The scheduled node
broadcasts its slot assignment to one-hop neighbors. Then these one-hop neighbors broadcast this 
information to update
two-hop neighbors. The algorithm proposed in \cite{DTSS}, called DTSS, provides a unified slot 
scheduling scheme for unicast, 
multicast and 
broadcast modes of transmission.
The DTSS algorithm assumes that all the nodes 
are synchronized with respect to a global time reference, before running the scheduling algorithm,
and also each node knows its set of 
intended receivers. 
Finally, a classification of different slot scheduling algorithms based on problem setting, problem 
goal, 
type of inputs and solution techniques, 
can be found in \cite{survey_TDMA}.

\section{Proposed Two-Phase TDMA Scheduling}
\noindent \textbf{Phase 1} \\
  Generate a TDMA schedule quickly using contention-based channel access mechanism. The 
 generated schedule in this phase need not have to be very efficient in terms of schedule length. For this phase, we propose
 a randomized and distributed TDMA scheduling algorithm (RD-TDMA) based on graph colouring approach. 
 A major advantage of RD-TDMA algorithm over existing TDMA-scheduling algorithms, is the multifold reduction in scheduling-time. This 
 is because, the static TDMA-scheduling algorithms typically use heuristic based approach (greedy approach) for graph colouring 
 which is inherently sequential in nature. On the other hand,  in the RD-TDMA algorithm, all the nodes can concurrently select their slots  using probabilistic approach.
%
%
%
\vspace{5pt}

\noindent \textbf{Phase 2} \\
Iteratively reduce the schedule length of the schedule generated in phase 1, using 
 TDMA-based channel  access. Note that, in this phase, we can use TDMA-based channel access  using  TDMA schedule 
 generated in phase 1. The process of schedule-length reduction is designed 
 in a manner such that it can be terminated after  the execution of arbitrary number of iterations, and still be left 
 with a valid schedule. This phase provides the flexibility to trade-off between schedule length and scheduling time. 
 For this phase, we have proposed a distributed schedule length
 reduction (DSLR) algorithm  which is deterministic in nature as opposed to the probabilistic nature of  RD-TDMA  algorithm proposed for  phase 1. 
 The basic idea behind the DSLR algorithm is fairly straightforward. In order to reduce the schedule length, all the nodes in the network move to  another slot with Id less than the Id of the slot currently occupied by them, without violating the conflict-free property 
 of the input schedule. The real challenge lies in the implementation of DSLR algorithm in parallel and distributed manner, such that
 no two  nodes which are   in  two-hop neighborhood of each other  simultaneously move to the same slot.
 
 \section{Phase 1: Distributed TDMA Scheduling}


The TDMA slot scheduling problem can be formally defined as the problem of assignment of a time 
slot to each node, such that if any two nodes are in conflict, they do not take the same time slot. Such an assignment is 
called a feasible TDMA schedule. Two nodes are said to be in \textit{conflict} 
if and only if  the transmission from one node causes an interference at any of the receiver of the 
other node. For example, in case of broadcast scheduling, a node cannot take a slot, if it is taken by any of 
its two-hop neighbors. A detailed list of conflict relations in wireless networks can be found in \cite{unified}.
In this paper, we assume the broadcast mode of communication to 
describe the proposed scheme, 
but it can be easily extended to unicast and multicast transmission modes.
Timeline is divided into fixed size frames and each frame is further divided into fixed number of 
time slots, $S$, called schedule length. 
Time slots within a frame are numbered from 1 to $S$, assuming that $S$ is 
sufficiently large enough to handle assignment for all possible input graphs. 
%
%
%

Now, we discuss the basic idea of the  proposed RD-TDMA algorithm followed by its description in detail.
Subsequently, we propose an optimization of RD-TDMA algorithm to achieve faster convergence, by dynamically updating the slot probabilities,
with which the nodes try to select different slots. 
The set of data structures and definitions that we have used to describe the RD-TDMA algorithm are collectively given 
in Table 1. \\

\noindent A. \textit{Overview of RD-TDMA Algorithm} \\
The basic idea of the proposed algorithm is as follows. For each slot $s$ in a frame, each node $i$ 
checks 
whether 
it can take the slot $s$, 
by broadcasting a  request  message  with slot probability $\boldmath p_{i,s}$. The value of $\boldmath p_{i,s}$ depends upon 
the remaining number of free slots in the frame, 
currently known at node $i$. When a node $j$ receives a request message from node $i$ for slot $s$, 
it 
grants the same 
to node $i$ if
it is not trying for the same slot or it  has not already granted the slot to some other node. If node $i$ receives grant 
from all its one-hop neighbors in $N_i$, it assigns the time slot $s$ to itself; otherwise, it leaves the slot,
as soon as it receives a reject message  from one of the nodes in $N_i$ and repeats the above process all over again. 
Once a slot is assigned to a node $i$, it informs the same to its neighbors by periodically broadcasting an indicate
message.
This would enable the neighboring nodes of node $i$ to leave the slot and try other slots with higher probability. Furthermore, the nodes 
in 
$N_i$ also 
propagate this information to their neighbors through their own transmissions so that two-hop neighbors of
$i$, $N2_i = (\cup_{j \in N_i} N_j) \cup N_i $ - \{i\}, are also informed. \\

%

 \begin{figure}[t]
\centering
\includegraphics[scale=.30]{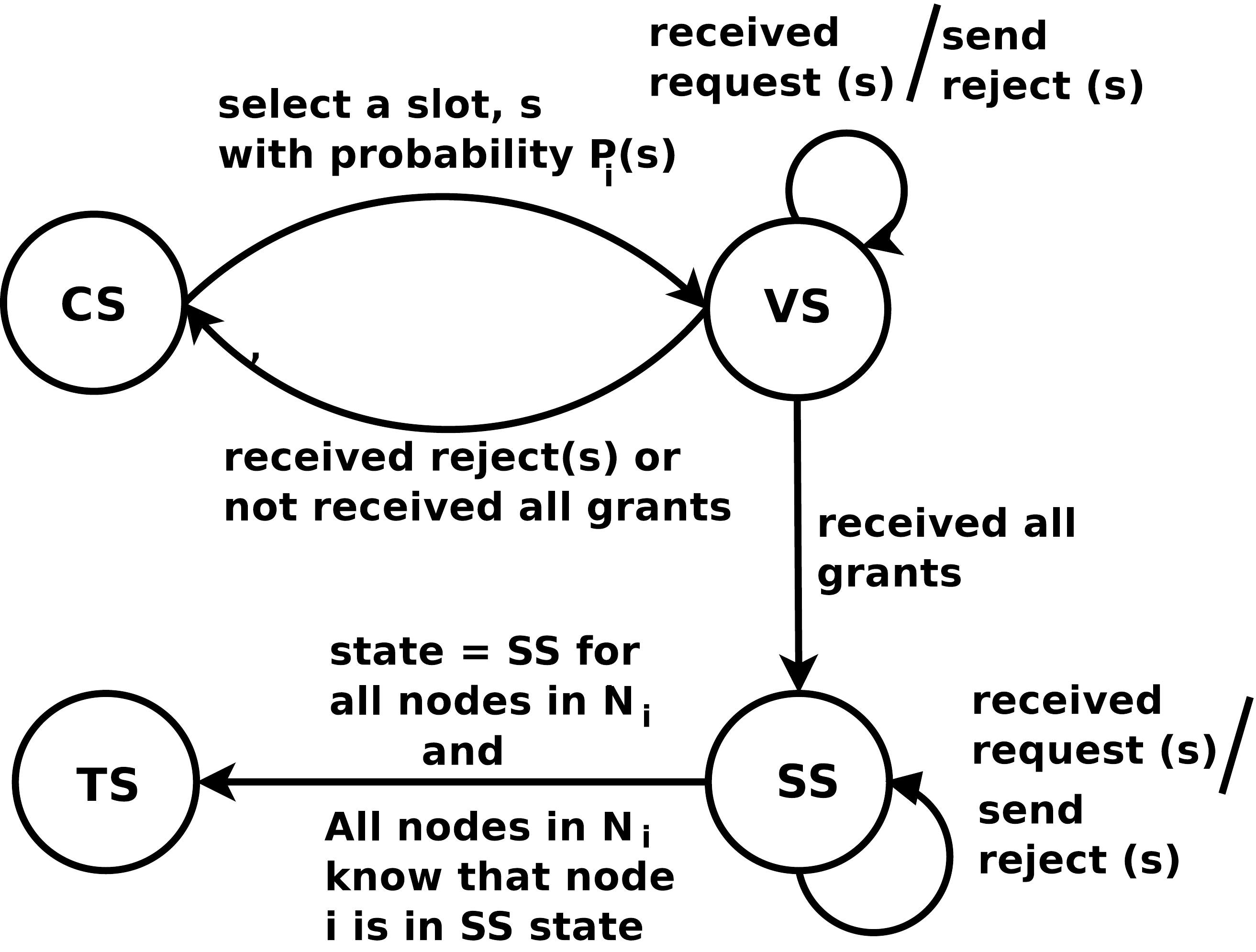}
\caption{The state transition diagram of a node $i$. 
}
\label{fig:FSM}
\end{figure}


\noindent B. \textit{Detailed Description of RD-TDMA}

Initially, each node $i$ sets $\boldmath{p_{i,s}} = \frac{1}{S}, where ~ 1 \le s \le S$. 
The value of $\boldmath{p_{i,s}}$ for each slot $s$ at a node keeps changing  during the execution of 
the algorithm, and it depends upon the corresponding slot-probabilities of the other nodes in the set $N2_i$. 
The summation of slot-probabilities at a node $i$, is always 1, i.e., $\sum_{s=1}^{S} p_{i,s} = 1$.

Each node contending for a time slot, passes through several states. 
Figure \ref{fig:FSM} shows the state transition diagram for a node $i$. 
Initially, a node \textbf{A} enters the \textbf{contention-state} (\textbf{CS}), where it randomly selects a slot 
$s$ as per the probability distribution defined by vector $\boldmath{P_A}$.
%
%
%
After selecting  slot $s$, node \textbf{A} 
enters the \textbf{verification-state} (\textbf{VS}). In \textbf{VS} state,  it  waits for $s$ time units (actual time depends upon the 
underlying data rate), and broadcasts a request (REQ) message. 
A random delay before transmitting the REQ message in  \textbf{VS} state,
avoids the collision between REQ messages, simultaneously  transmitted by various nodes in a proximity. 
Remaining  collisions are assumed to be handled by the underlying MAC layer. 
When a node \textbf{B} $\in N_A$  receives a REQ message for slot $s$ from node \textbf{A}, it grants slot 
$s$ to \textbf{A}, and informs the same by
piggybacking the grant information in its own subsequent messages
using a vector,  called \textbf{grant-vector} (\textbf{GV}),  only if any of the following 
conditions are met. The vector \textbf{GV} in the REQ message sent by node $i$ is  the same as the local vector \textbf{gV} at node 
$i$.

\noindent 1) Node \textbf{B} has not  granted slot $s$ to any other node. \\
2) Node \textbf{B} is not in \textbf{V}S state with respect to slot $s$ i.e. it has not sent a request for slot $s$ and waiting for grants. \\
3) Node \textbf{B} had granted slot $s$ to node \textbf{C}, but subsequently it received request for 
another slot $u$
 from node \textbf{C}.
\begin{table*}[t]
  \centering
   \begin{tabular}{|c|p{9cm}|}
    \hline
    \textbf{Notation} & \textbf{Description} \\
    \hline	
  $\mathbf{N_{i}}$ &  Set of one-hop neighbors of a node i  \\ 
  $\mathbf{N2_i}$ & Set of two-hop neighbors of node i  $(\cup_{j \in N_i} N_j) \cup N_i   - \{i\} $ \ \\
\textbf{gV} & A vector of size $1 \times S$, where \textbf{gV(s)} at node $i$ contains the node id which has been granted slot $s$ 
by the 
node $i$\\
\textbf{oV} & A vector of size $1 \times S$, where \textbf{oV(s)} at node $i$ indicates that the slot $s$ is occupied if 
\textbf{oV(s)} 
$=1$\\
\textbf{rgV} & A vector of size $1 \times |N_i|$, where \textbf{rgV(k)}  $=1$ at node $i$  indicates that node $i$ has received grant from 
a node j 
such that $order(j) = k$. Here, $``order" $ is a function which maps the ID of  all the nodes in  $N_i$ to a unique number between $1$ and 
$|N_i|$. \\  
\textbf{indV} & A vector of size $1 \times |N_i|$, where \textbf{indV(k)}  $=1$ at node $i$ denotes that node $i$  is aware that ``a node j 
such that
$order(j) = k$, 
knows that node $i$ has occupied a slot''.  \\
 $\mathbf{P_i}$ & A slot-probability vector of size $1 \times S$,  where $P_i(s) = \boldmath{p_{i,s}}$ \\ 
 $\mathbf{slot_i}$ &  Id of the slot taken by node $i$ \\
       \hline
 \end{tabular}
 \caption{The set of data structures and definitions used in the implementation of RD-TDMA algorithm}
  \label{table:complexity}
\end{table*}

 After granting  the slot $s$ to node \textbf{A}, node 
\textbf{B} leaves the slot $s$ temporarily by setting 
$P_B(s) = 0$, until it 
 receives another REQ message from node \textbf{A}  
 for a slot $u$ other than slot $s$.

If node \textbf{B} receives a REQ message from node \textbf{A} for a slot $u$ other than slot $s$, 
which it 
has already granted to node \textbf{A}, it revokes the grant of 
 slot $s$ to A and assigns a new value to $P_A(s)$. Moreover, it can  grant  slot $u$ to node
\textbf{A} as per the three conditions 
 given above.  In case node \textbf{B} receives a REQ message from \textbf{A} for an already granted slot to 
\textbf{A}, it will
 simply ignore it.
  If REQ message transmitted by node \textbf{A} in \textbf{VS} state is not received at one or more nodes in  $N_A$,
  either due to channel error or collision, node \textbf{A} will  eventually not receive any grant or reject 
  from those nodes. In this case, it will retransmit the  REQ message  for slot $s$, after waiting for 
a time uniformly distributed between 0 and $S$. 
 While in \textbf{VS} state,  if  node 
\textbf{A} receives a REJECT message  or does not receive the grant from each of the nodes in $N_A$ within a $MAX\_ATTEMPTS$ number of 
transmission of REQ message,  it goes back to \textbf{CS} state.

After receiving grant from every node in $N_A$ for slot $s$, node \textbf{A}
 enters the \textbf{scheduled-state} (\textbf{SS}), and it sets the 
 slot-probability $P_A(s) = 1$, $slot_i = s$ and $P_A(u) = 0, \forall u \ne s $.
 After entering \textbf{SS} state,
 node \textbf{A} broadcasts an indicate (IND) message to the nodes in $N_A$, informing them that it has taken the slot.  
When node \textbf{B} receives an IND message from node \textbf{A} with respect to slot $s$, it 
will leave the slot $s$ permanently by setting $P_B(s) = 0$.
 Furthermore, it will convey the same information to the nodes in
 $N_B$ through its own subsequent REQ messages.
 This is achieved  by adding a  bit vector field, called \textbf{occupied-vector} (\textbf{OV}),  in the REQ messages, 
specifying that a particular slot is already 
 taken by a node in $N_B$ (in this case, node \textbf{A}), if the corresponding bit is set. 
Furthermore, when node  \textbf{A}  receives an \textbf{OV} vector  from node \textbf{B} s.t.  \textbf{OV}$(slot_A) = 1$, it sets $indV(order(B)) = 1$.  This tells node  \textbf{A}  that  node  \textbf{B} is aware about the fact that node   \textbf{A}  has occupied  a slot.
The idea of adding   an \textbf{occupied-vector} (\textbf{OV}) field   in the REQ message of node \textbf{B}  instead of   simply relaying    IND message  that  node \textbf{B} received from node  \textbf{A},   mitigates the problem of 
\textbf{broadcast-storm} considerably. 
 The vector \textbf{OV} in the REQ message sent by any node $i$ is  the same as the local vector \textbf{oV} at node 
$i$. When node \textbf{C}  
receives 
a REQ message from node \textbf{B}, it will also set the slot-probability to 0 for those slots where the corresponding bits in vector 
\textbf{OV} are set.

Finally, from \textbf{SS} state, node \textbf{A} would enter the \textbf{terminate-state} (\textbf{TS}), if every node in  $N_A$ is in
\textbf{SS} state and all of them would know that the node \textbf{A} is also in \textbf{SS} state.  This situation can simply be tested by checking whether all the bits in  the vector \textbf{indV} are set. 
In \textbf{TS} state, the execution of RD-TDMA algorithm would stop, and  the node would not  transmit any further messages. \\


\noindent C. \textit{Dynamic Slot Probabilities} \\
In order to achieve faster convergence, it is desirable that node $i$ should update its slot probabilities, instead of 
always trying a slot $s$ with $p_{i,s} = 1/S$. 
For example, if all the nodes in $N2_i$ are in \textbf{SS} 
state, then  node $i$ can take any free slot. 
In another situation, if all the nodes in $N2_i$ have already left  slot $s$, (since their two-hop neighbors are 
in \textbf{SS} state with respect to slot $s$), then node $i$ can take  slot $s$ with probability one. 
We propose a dynamic slot probability assignment mechanism in 
which node $i$ shares its slot-probability vector $P_i$, with all the nodes in $N2_i$, by explicitly 
transmitting an advertisement message.
We define the term \textbf{probability-budget}, $budget_{i}^{s}$ at a node $i$ for slot $s$,  as 
$budget_{i}^{s} = 1- \left(\left(\sum_{j\in N2_i} P_{j}(s)\right) + P_{i}(s)\right)$, where $P_{i}(s)$ is the 
slot-probability of node $i$ for slot $s$.

The negative value of $budget_{i}^{s}$ indicates that all the nodes in $N2_i$ including node $i$ are together 
trying for slot $s$ with probability
more than one. This will lead to higher level of  contention among neighboring nodes, to take slot $s$, and 
therefore, the probability that 
some node will succeed, would be lesser.  
%
On the other hand, the positive value of $budget_{i}^{s}$ indicates that all the nodes in $N2_i$ including node 
$i$ are 
together trying for slot $s$, 
with probability less than one. This will lead to the situation where the nodes are not trying with sufficient 
probability to get  slot $s$, and therefore, the probability that some node will succeed, is again less. 
Note that $\left(\sum_{j\in N2_i} \sum_{s=1}^{S}P_{j}(s)\right) = |N2_i|$, and therefore, the  value 
of $|budget_{i}^{s}| \ne 0$ implies that 
the distribution of probability sum of all neighboring  nodes,  i.e., $\sum_{j\in N2_i} P_{j}(s)$,  is not uniform across 
all slots.

%
  
%
  
%
The above discussion suggests that every node should try to maintain 
$budget_{i}^{s} \approx 0$ for every slot. In order to achieve the above goal, we propose a method for updating the 
slot-probability 
vector $P_i$ at each node $i$ as follows. 
When a node $i$ receives an update of slot-probability from some other node, it calculates the $budget_{i}^{s}$ 
for 
every slot and updates vector $P_i$, using the following equation. 
\abovedisplayskip=2pt
\belowdisplayskip=2pt
\begin{align}
\label{eq_budget}
 P_{i}^{new}(s) =
\left\{
\begin{array}{ll}
 0,  ~~~~~ \text{if } \exists j \in N2_{i} \text{ s. t. }  P_j(s) = 1 \\ 
 P_{i}(s) +  \mathcal{K}*budget_{i}^{s},  ~ \text{otherwise} 
\end{array}
\right. 
\end{align}
, where $\mathcal{K}$ is a constant, $\frac{1}{S} \le \mathcal{K} \le 1$. Node $i$ does not 
compensate complete probability budget 
at once, but it compensates a 
fraction of it, determined by the constant $\mathcal{K}$.
This is because, if most of the neighbors of node $i$ are also the 
neighbors of each other, then they will also update their slot-probabilities using the same budget. This could lead to  
an 
unstable situation, where the nodes keep on updating (increasing and decreasing) their slot-probabilities  instead of 
converging to a certain value. 

\section{Phase 2: Distributed Schedule Length Reduction (DSLR) Algorithm}
In this section, we first give the assumptions and a couple
of definitions for the DSLR algorithm, followed by data structures
maintained at each node $i$, to implement the DSLR algorithm
in a distributed manner. Then, we present the proposed DSLR
algorithm in detail.

\subsection{Assumptions, Definitions and Data Structures}
The DSLR algorithm takes the TDMA schedule of length $F$, generated in phase 1  as   input, and produces another
schedule of length $\le F$, as the output of the algorithm. 
During the execution of the DSLR algorithm, each 
node will possibly move to a  slot whose  Id is  lesser than the Id of its current slot. We use the term $slot_i$ to
denote the Id of the current slot of node $i$. However, throughout the execution of the DSLR algorithm,
a sensor node will transmit only in the slot defined by the initial TDMA schedule. Hence, the
notion of current slot Id of a node $i$ ($slot_i$) is logical instead of the physical slot where node $i$
will actually transmit during the execution of DSLR algorithm.

In the following, we introduce some definitions, which are
required to describe the proposed DSLR algorithm.

%
%
%
%

\noindent\textbf{Definition 1. }\textit{The receiver set $R_i$ of a node $i$ is defined as the
set of intended receivers of node $i$}.

\noindent\textbf{Definition 2. }\textit{The sender set, $S_i$,  of a node $i$ is defined as the
set of intended senders of node $i$}.

It is to be noted that  $R_i \subseteq N_i$ and  $S_i \subseteq N_i$. The size of the set $R_i$, $|R_i|$ depends upon the type of 
communication, viz., unicast, multicast or broadcast transmission.
Note that, in  WSNs, all nodes cooperate for a single task, and  only one application 
run at any given time. Hence, the information related to the set of receivers can be easily made available 
at MAC layer. 


\noindent\textbf{Definition 3. }\textit{A slot $s$ is said to be free at node $i$ if,
$\forall j \in ((\cup_{k \in N_i} S_k) \cup (\cup_{k \in R_i} N_k)), slot_j \ne s$.}

\noindent\textbf{Definition 4. }\textit{
The first-free slot, $FF_i$ of a node $i$ is defined
as the slot with minimum slot Id out of all free slots available 
at node $i$. If no free slot is available, the value of $FF_i$ is set to $0$.} 

%

\noindent\textbf{Definition 5. }\textit{
Two nodes $i$ and $j$ are said to be in conflict, if
the transmission of one node interferes at one of the receivers
of the other node, i.e., $N_i \cap R_j \ne \phi \vee N_j \cap R_i \ne \phi$. When two nodes are in conflict, they cannot take the same slot. 
}

In the following, we give a brief description of data
structures, which are maintained at each node i in the DSLR
algorithm.

\begin{itemize}
 \item A one dimensional vector $SV_i$, of size $|N_i| \times 1$, where
the value of $SV_i[order(j)] = s (\ne 0)$ specifies that
 $slot_j = s$. If $SV_i[order(j)] = 0$, it means that node
$i$ does not know the slot occupied by node $j$. Here,
$order(j)$ is a function which maps the Ids of every
node $j$  in $N_i$ to a unique number between $1$ and $|N_i|$.

\item A one dimensional vector $RV_i$, of size $F \times 1$, where
$0 \le RV_i[s] \le 3$, specifies the status of slot $s$ in terms
of its occupancy by the nodes in $N_i$, and $F$ is the initial schedule length.

\item A vector of vectors $NRV_i$, of size $|N_i| \times 1$, where the
value of $NRV_i[order(j)]$ specifies the vector $RV_j$.

\item A one dimensional vector $FFV_i$, of size $|N_i| \times 1$,
where the value of $FFV_i[order(j)] = FF_j$.

\item A one dimensional vector $maxSV_i$, of size $F \times 1$,
where the value of $maxSV_i[s]=u$, specifies that $u =
 max(slot_j : j \in N2_i \wedge FF_j = s)$. The value of
    $maxSV_i[s]$ is set to $0$, if $\forall j \in N2_i, FF_j \ne s$.

\end{itemize}


\subsection{Overview of DSLR Algorithm}
The basic idea behind the proposed DSLR algorithm is
fairly straightforward. In order to reduce the schedule length,
each node  in the network moves to another slot (if available) whose  Id is 
less than the Id of the slot currently occupied by it. The
real challenge is to perform this operation in parallel and
distributed manner, such that, two or more conflicting nodes
do not move to  the same slot, making the TDMA-schedule non-
feasible. In order to do this, nodes in the network, calculate
the status (free or occupied) of all the slots by exchanging
the messages with their neighboring nodes, and then move
to the \textit{first-free} slot. Additionally, while moving from one
slot to another, the nodes need to  ensure that no node in their two-hop neighborhood simultaneously moves to the same slot.
Note that, two nodes $i$ and $j$ in a two-hop neighborhood can
simultaneously go to their \textit{first-free} slots if the slot Ids of
their \textit{first-free} slots are not the same. In this sense, the DSLR
algorithm executes in parallel not only at different parts of
the network which are geographically apart, but also within
a two-hop neighborhood.

\subsection{Detailed Description of the DSLR Algorithm}

The DSLR algorithm executes in synchronized rounds,
where each round consists of four consecutive frames. Each
node $i$ periodically transmits a HELLO message  in
every frame of a round at  the  slot  as per the input TDMA schedule.

In the first frame of a round, the HELLO message sent
by node $i$  contains the value of its current slot Id ($slot_i$). 
Initially, the value of  $slot_i$ and the slot where node $i$ transmits the HELLO message
are  same. However, during the execution of the DSLR algorithm,  node $i$ may move to a new slot whose Id is  less than the value of $slot_i$,  and therefore, the value of  $slot_i$ and the Id of  the  slot for HELLO message transmission, may differ.

 If a node $j$ receives a HELLO message
 sent by node $i$ in the first frame of a round containing $s$ as the current slot Id, then
it can directly set $SV_j[order(i)] = s$.
 The values of $SV_i$, $RV_i$ and $FF_i$ are reset to $0$ at
the beginning of the first frame of every round. Additionally,
a copy of $SV_i$ is stored in the vector $SV_i^{old}$ before resetting
$SV_i$ to $0$ in the following manner.
\vspace{-1pt}
\[
 SV_i^{old}[order(k)] = SV_i[order(k)], \forall k : SV_i[order(k)] \ne 0
\]

In the second frame of a round, the HELLO message sent
by node $i$ contains the vector $RV_i$. The value of vector $RV_i$
can be calculated with the help of vector $SV_i$, that has been
populated during the first frame of the same round. Each
element of the vector $RV_i$ can take values between $0$ and $3$
inclusive, and it is calculated as follows.

\begin{enumerate}
\item for each  $j \in S_i$, $\text{ if } SV_i[order(j)] = s \wedge s \ne 0,    \text{ then } RV_i[s] = 1$.

\item   for each  $j \in N_i -  S_i$,  $\text{ if } SV_i[order(j)] = s \wedge s \ne 0,  \text{ then }   RV_i[s] = 2$. 

\item   for each  $j \in N_i$,  $\text{ if } SV_i[order(j)] = 0 \wedge SV_i^{old}[order(j)] = s \wedge s \ne 0,     \text{ then } RV_i[s'] = 3, \forall s' \le s$.

\item For all other slots s, which are not modified in steps (1), (2) and (3) above, set $RV_i [s] = 0$.
\end{enumerate}
\vspace{.1cm}

%
%
%
%
%
%
%
%
%

The value of $RV_i[s] = 0$ informs to the receivers that the
slot $s$ is free for all nodes in $N_i$. $RV_i[s] = 1$ indicates that slot
$s$ is not free for any node in $N_i$. $RV_i[s] = 2$ indicates that slot
$s$ is not free for the nodes in $S_i$ and free for the remaining
nodes in $N_i$. Finally, $RV_i[s] = 3$ indicates that the status
of slot $s$ is not known to node $i$. The reason behind setting
$RV_i[s] = 3$ for all slots with $Id \le SV_i^{old}[order(j)]$ is that,
once node $i$ does not receive a HELLO message from node
$j$, the node j may have moved from its old slot to another
slot. But, as we will see later, a node can only move to a slot
with lesser slot Id than its current slot Id, and thereby, it is not
required to set $RV_i[s] = 3$, $\forall s > SV_i^{old}[order(j)]$. On receipt
of a HELLO message from a node $j$ in the second frame of a
round,  node $i$ would set $NRV_i[order(j)] = RV_j$. If node $i$
does not receive a HELLO message from node $j$ in the second
frame of a round due to any reason such as frame loss, then
the value of $NRV_i[order(j)][s] = 1, 1 \le s \le F$.

In the third frame of a round, the HELLO message sent
by node $i$ contains $FF_i$. The value of $FF_i$ is calculated with
the help of vector $NRV_i$, that has been populated during the
second frame of the same round. Node $i$ considers a slot $s$
as free, if and only if the value of $RV_j[s] = 0, \forall j \in R_i$ (i.e.
$\forall j : i \in S_j$), and $RV_k[s] = 0$ or $2$, $\forall k \in N_i - R_i$. On receipt
of a HELLO message from a node $j$ in the third frame of a
round, the node $i$ would set $FFV_i[order(j)] = FF_j$.

In the fourth frame of a round, the HELLO message sent
by node $i$ contains $maxFree_i$. Here, $maxFree_i$ is a list of
ordered pairs $(s1, s2)$, such that $s1$ is a \textit{first-free} slot with
respect to at least one node in $N_i$ and $s2 = max(slot_j : j \in N_i \wedge FF_j = s1)$. The value of 
$maxFree_i$ can be calculated
with the help of vector $FFV_i$, that has been populated during
the third frame of the same round. On receipt of a HELLO
message from a node $j$ in the fourth frame of a round, node
$i$ would set the value of $maxSV_i[s1] = s2$, for each $(s1, s2)$
pair present in the received HELLO message, if $s1 < s2$ and
$maxSV_i[s1] < s2$.

Finally, in the beginning of a round, node $i$ with $slot_i = s$
can go to its \textit{first-free} slot, $FF_i$ if all two-hop neighbors
of node $i$ with the same \textit{first-free} slot as that of node $i$,
have occupied the slots with Id less than $s$. In other words,
$(j \in N2_i \wedge FF_j = FF_i) \Rightarrow slot_j < slot_i$, given $FF_i \ne 0$.
That is, $maxSV_i[FF_i] < slot_i$.



\section{Correctness of RD-TDMA and DSLR Algorithms}

\subsection{RD-TDMA  Algorithm}
\begin{theorem}
For any two nodes $i$ and $j$, if $j \in N2_i$, then both the nodes cannot be in \textbf{SS} state 
with respect to the same slot. 
\end{theorem}

\begin{proof}
 In order to enter \textbf{SS} state,  both the nodes $i$ and $j$ have to  get grants from all the 
nodes in $N_i$ and 
$N_j$ 
respectively. This is because
$j \in N2_i \Rightarrow j \in N_i \vee  \exists k \in N_i : j \in N_k$. 
When $j \in N_i$,  node $j$ will not take the slot, if it has already granted the slot to 
node 
$i$, and 
therefore both
cannot be in \textbf{SS} state with respect to the same slot. If $j \notin N_i$ then  $\exists k 
\in 
N_i$ such that  $j \in N_k$.
Due to symmetric channel  relationship, $j \in N_k \Rightarrow k \in N_j$. In order to be in 
\textbf{SS} state with respect to a 
slot $s$, 
both $i$ and $j$ need to get the grant from node $k$ for the slot $s$, 
which is not possible, because at a time, node $k$ will grant slot $s$, either to node $i$ or to 
node 
$j$ but not to both. Note that with packet loss, neither of them may be able to enter  SS state.
\end{proof}

\begin{theorem}
  Every node in the network will eventually enter \textbf{SS} state in a finite number of attempts.
\end{theorem}

\begin{proof}
Let $P_i^{succ}$ be the probability of a successful packet transmission by a node $i$.
Let 
$P_i^{SS}(s)$ be the 
transition 
probability, defined as the probability that node $i$ will enter \textbf{SS} state from \textbf{VS} 
state with respect to a free 
slot $s$. Since node $i$ is in \textbf{VS} state, $ 0 < p_{i,s} \le 1$. 
A slot $s$ is free for node $i$,
if no node in $N2_i$ is already in SS state with respect to slot $s$, and no node in $N_i$ has 
already granted 
slot $s$ to some other node . 
In other words, slot-probability $P_j(s) < 1, \forall j \in N2_i$.
In this case, we need to show 
that the transition probability $P_i^{SS}(s) > 0$, irrespective of the slot-probabilities of other 
nodes in the 
network, so that node $i$ will eventually enter SS state in a finite number of attempts.
The value of $P_i^{SS}(s)$ is $ p_{i,s}*\left(\prod_{j \in N2_i} (1- P_j(s))\right)* (P_i^{succ}) * 
\prod_{k \in N_i}  
P_k^{succ}$, 
i.e., the probability that no node in $N2_i$ simultaneously selects $s$ along with node $i$, 
and node $i$ and all nodes in $N_i$ successfully transmit their REQ message.
Since $p_{i,s} > 0$, $P_j(s) < 1, \forall j \in N2_i$,
$P_i^{succ} > 0$ and  $P_k^{succ} > 0,   \forall k \in N_i$,  the inequality $P_i^{SS}(s) > 0$ would always be satisfied. Finally, the  
loss of 
any message (REQ, IND, REJECT) due to collision or channel impairment, would not affect the 
convergence of  
the algorithm. But, this will surely increase the convergence time. 
\end{proof}

\subsection{DSLR Algorithm}
In this section, we first prove the correctness of DSLR
algorithm, i.e., the DSLR algorithm always generates a feasible
schedule at the end of every round, if it initially starts with a
feasible schedule. Thereafter, we will show that there exists an
upper bound on the schedule length generated by the DSLR
algorithm.

\begin{lemma}
The slot $FF_i$ calculated by node $i$ in the third frame of a round is free according to definition 3.
\begin{proof}
According to the definition 3, a slot is free for node $i$ if  $\forall j \in (( \cup_{k \in N_i}  S_k) \cup ( \cup_{k \in R_i} N_k)), slot_j \ne s$. In the DSLR algorithm, node $i$ considers slot $s$ as free iff the value of $RV_j[s] = 0, \forall j \in R_j$ and $ RV_j[s] = 0 ~ or~  2, \forall j \in N_i - R_i$.
\end{proof}
\end{lemma}

\begin{theorem}[\textbf{Correctness without packet loss}]{\textit{
 In case of no packet loss, the DSLR algorithm 
always generates a feasible schedule, at the end of each round, 
if it starts with a feasible schedule, i.e., for any two nodes $i$ 
and $j$ in the network $(N_i \cap R_j \ne \phi) \vee (N_j \cap R_i \ne \phi) \Rightarrow
          slot_j  \ne slot_i.$}} 
\begin{proof}
 In order to prove the theorem, we need to show that 
the following condition always holds, even after the movement 
of node $i$ from its current slot ($slot_i$) to its \textit{first-free} slot ($FF_i$). 
\vspace{-1pt}
\[
 (N_i \cap R_j \ne \phi) \vee (N_j \cap R_i \ne \phi) \Rightarrow  slot_j \ne FF_i ~~~~~~~  (1)
\]

From the definition of free slot  (Def 3), we know that
$\forall j \in ((\cup_{k \in N_i} S_k) \cup ( \cup_{k \in R_i} N_k))$, $slot_j \ne FF_i$. 
Therefore, In order to prove that condition (1) always holds, it is sufficient
to show that,
\[
 (N_i \cap R_j \ne \phi) \vee (N_j \cap R_i \ne \phi) \Rightarrow j \in (\cup_{k \in N_i} S_k ) \cup (\cup_{k \in R_i} N_k )
\]
\vspace{.1cm}
Now $N_i \cap R_j \ne \phi$ implies that $\exists k \in N_i$ such that $j \in S_k$.
Therefore, $j \in (\cup_{k \in N_i} S_k)$. Similarly, ($N_j \cap R_i \ne \phi$) implies
that $ \exists k \in R_i$ such that $j \in N_k$. Therefore, $j \in (\cup_{k \in R_i} N_k)$.
Hence, $(N_i \cap R_j \ne \phi) \vee  (N_j \cap R_i \ne \phi) \Rightarrow j \in (\cup_{k \in N_i} S_k) \cup
(\cup_{k \in R_i} N_k)$.
\end{proof}
\end{theorem}

We give two other theorems to extend the scope of theorem 1 for the case when, 
due to erroneous wireless channel, the
HELLO message transmitted by node $i$ may not be received
by all the nodes in $N_i$.

\begin{theorem}[\textbf{Correctness with packet loss}] {\textit{
  If $s$ is the value of $FF_i$ calculated by node $i$ at
the beginning of the third frame of a round, then $s$ is a free
slot (Def. 3), with respect to node $i$, i.e., $(FF_i = s \ne 0) \Rightarrow$
slot $s$ is free.}}

\begin{proof}
 If the node i does not receive one or more HELLO
messages in the first or second frame of a round, due to
 erroneous channel, it is hard to say about a few slots, whether
they are free or not. This is because, some nodes in $N2_i$ might
have moved to their first-free slot at the beginning of the same
round.

If node $i$ does not receive a HELLO message from
node $j$ in the second frame of a round, then the value of
$NRV_i[order(j)][s] = 1, 1 \le s \le F$, and consequently,  $FF_i = 0$.

Now, consider the situation when the node $i$ has received
all the HELLO messages in the second frame, but say a node
$j \in N_i$ did not receive the HELLO message from some other
node $k$ in $N_j$, in the first frame of the same round. In that
case, the node $j$ would have set $RV_j[s] = 3$ for all slots with
Id less than or equal to the slot Id of last known slot Id $u$
of node $k$. Finally, node $i$ considers a slot $s$ as free, if and
only if the value of $RV_j[s] = 0, \forall j \in R_i$ and $RV_j[s] = 0$ or
$2$, $\forall j \in N_i -  R_i$. Hence, in this case, $FF_i > u$, if the HELLO message
sent by a node $k$ in $N2_i$ is lost, in the first frame of a round,
where u is the last known slot occupied by node $k$. Hence,
in case of packet losses, the above argument proves that the
node i would not assume a slot as free, if it is unsure about
the status of the slot.
\end{proof}

\end{theorem}

Next, we discuss the performance of the DSLR algorithm in
terms of its capability to reduce the length of a given schedule.
In order to do so, we define \textit{interference graph} $G = (V, E)$,
where V is the set of nodes in the WSN, and E is the set
of edges. An edge $e = (i, j)$ exists if and only if $N_i \cap R_j \ne \phi \vee 
N_j \cap R_i \ne \phi$, i.e., an edge $e$ exists between node $i$ and
$j$, if and only if node $i$ and node $j$ cannot take the same slot.
Further, the \textit{Interference degree} of the interference graph G is
defined as the maximum of the degrees taken over all vertices
of the graph, and it will be denoted by $\Delta$.

\begin{lemma}
  After executing the DSLR algorithm for sufficiently
large number of rounds, no node in the network has a 
free slot whose  slot Id is less than its own slot Id.

\begin{proof}
 We prove this by contradiction. Let $i$ be a node
for which $s$ is a free slot such that $s < slot_i$. If all the two-
hop neighbors of node $i$ with \textit{first-free} slot as $s$ have occupied
the slots with Id less than $slot_i$, then the node $i$ can move to
slot $s$, which makes slot $s$ as occupied, and it contradicts our
assumption. If the $slot_i$ is not maximum, and there exist some
other node $j$ in $N2_i$ such that $slot_j > slot_i$, then again either
$slot_j$ would be maximum with respect to free slot $s$ in $N2_j$ or
some other node $k$ in $N2_j$ such that $slot_k > slot_j$.  If $slot_j$ is
maximum, then node $j$ can move to slot $s$ making it occupied
with respect to node $i$, which contradicts our assumption.
Finally, by giving the same argument repetitively, we can say
that, starting from node i, there exists a sequence of nodes,
$S = i, j1, j2 \dots jn$, with slot $s$ as the free slot, such that
$slot_i < slot_{j1} < slot_{j2} \dots < slot_{jn}$. In this case, the nodes
in the reverse order of sequence S, starting from node $jn$ to
node $i$ can move to slot $s$, eventually making slot $s$ occupied
for node $i$.
\end{proof}
\end{lemma}

The TDMA-schedule generated by greedy graph colouring
based distributed algorithms also maintains the same property
as mentioned in lemma 2. In the greedy approach, the exact
schedule depends on the order in which the nodes have been
coloured (provided slot). In our case, it solely depends upon
the input schedule provided to the DSLR algorithm. Another
way to see the similarity between the DSLR and the greedy
approach is that, in greedy approach, we gradually pick colours
one by one till all the nodes in the graph are coloured, whereas,
in the DSLR algorithm, we start with a coloured graph and
incrementally drop the colours one by one till no more colours
can be dropped, and still have the whole graph coloured.

\begin{theorem}[\textbf{Upper bound on SL}]{\textit{
The schedule length, SL generated by the DSLR
algorithm is always less than or equal to $\Delta + 1$.
}}
 
\begin{proof}
 Let there exist a node $i$ in the network such that
$slot_i > \Delta+1$. In this case, since no node can have degree more
than $\Delta$, as per the pigeon hole principle, there should exist at
least one slot s, less that $slot_i$, which has not been occupied
by any node adjacent to node $i$ in graph G, and therefore slot
$s$ is free with respect to node i. But this contradicts Lemma 1.
Hence, no node in the network can have slot Id greater than
$\Delta + 1$, and $SL \le \Delta + 1$.
\end{proof}
\end{theorem}

\begin{table*}[t]
  \centering
 
  \begin{tabular}{|c|p{10cm}|}
    \hline
    \textbf{Notation} & \textbf{Description} \\
    \hline	
      $t_i^{CS}(r)$ & Time spent by node $i$  in CS state at round $r$. \\
      $t^{req}$ & Time required to transmit the REQ message. \\
      $t_i^{w}(r)$ & Waiting time of node $i$ in VS state at round $r$, before transmitting the REQ message.  \\ 
      $t_i^{w1}(r)$ & Waiting time of node $i$ in VS state at round $r$, after transmitting the REQ message.  \\ 
      $t_{i}^{VS}(r)$ & Total time spent by node $i$ in VS state at round $r$, i.e., $t_i^{w}(r) + t^{req} + t_i^{w1}(r)$ \\ 
      $t_{i}^{SS}$ & Total time taken by node $i$ to reach SS state \\ 
      $R_i$ & Total number of rounds taken by node $i$ to reach SS state \\ 
      $R$ & Total number of rounds taken by all the nodes in the network to reach SS state \\ 
       $q_{i,s}(r)$ & The probability of node i  entering SS state with  slot $s$ in a round r.   \\
      $q_i(r)$ & The probability of a node i entering SS state, in  round  r  i.e.,  $\sum_{s=1}^{S} q_{i,s}(r)$ \\	  
      $q_i^{min}$ & Minimum value of $q_i(r) ~  \forall r$. \\
      $q^{min}$ & Minimum value of $q_1(r) ~  \forall r$. The subscript 1 has been omitted for sake of simplicity \\
      $\beta_{i}$ & The number of 1's in row i of matrix B. \\
      $\alpha_s$ &  The number of 1's in column $s$ of matrix B, excluding first row. \\
      $n$ &  Number of nodes in the network. \\
       \hline
 \end{tabular}
  \caption{\textbf{The  set of notations used in section \ref{complexity}}} 
  \label{table:complexity}
\end{table*}

\section{Complexity Analysis of RD-TDMA and DSLR Algorithms}
\label{complexity}
In this section, we analyze the runtime performance of RD-TDMA and DSLR algorithms in terms of time taken by the algorithms to perform 
TDMA scheduling and schedule length reduction respectively. 

\subsection{Runtime of RD-TDMA Algorithm}
Let $T$ be the time when all the nodes in the network reach SS state. Our goal is to calculate $E[T]$.
Table \ref{table:complexity} summarizes the set of notations used in this section.
Each node $i$ in the network runs the RD-TDMA in rounds till it finishes the slot selection, i.e., 
reaches the SS state. We denote the time duration between the events  when 
node $i$ enters CS state for the $r^{th}$ time  and  leaves the VS state by $t_i(r)$. A higher level 
behaviour of node $i$ while it is contending for a slot to select, along with 
the description of  various time delays, is described in the following. (also see Fig 
\ref{fig:analysis}). 
\begin{figure*}[t]
\centering
\includegraphics[scale=.55]{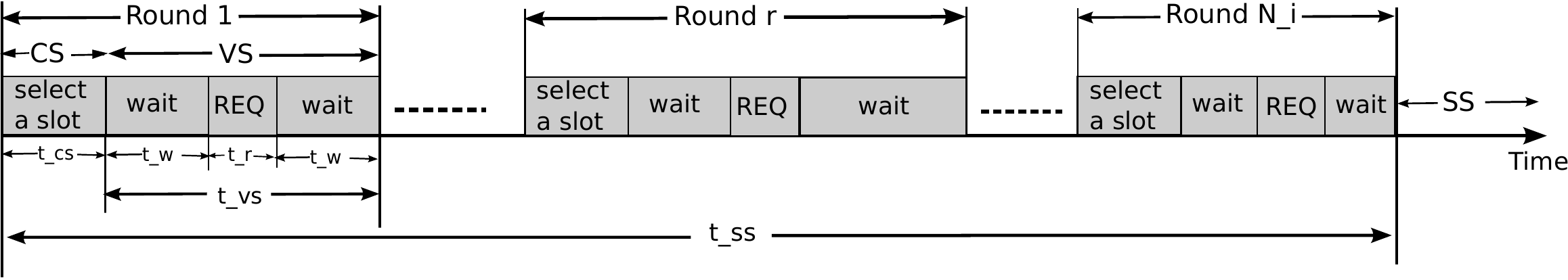}
\caption{Description of various time delays while a  node i is contending for a slot 
using RD-TDMA algorithm}
\label{fig:analysis}
\end{figure*}
\begin{itemize}
\item Initially, a node $i$ enters \textbf{CS} state, where it randomly 
selects a slot $s$, as per the probability distribution defined by vector $\boldmath{P_i}$. 
Let  $t_i^{cs}(r)$ be the time spent by  node $i$ during round $r$ in CS state, i.e., the time 
required by node $i$  to select a slot. We assume that the value of $t_i^{cs}(r)$ is fixed across 
all the sensor nodes in the network. The actual value of $t_i^{cs}(r)$ depends upon the 
underlying hardware (processor) that has been  used to execute the protocol instructions.
\item In \textbf{VS} state of round $r$, node $i$ first  waits for a random duration of time $t_i^w(r)$, which is
uniformly distributed between 0 and S, and then it transmits a REQ message for  slot $s$. 
The time required to transmit a REQ message $t^{req}$ depends upon the size of REQ message and 
 data rate for transmission. 
\item If node $i$ does not receive all grants or receives a REJECT message, 
withing  $S$ time units after transmitting the REQ message,  it goes back to CS state 
and restart 
the process again. 
After transmitting a REQ message,  the time spent by node $i$ in VS state of round $r$,  
$t_i^{w1}(r)$, is also a random variable. $t_i^{w1}(r) = S$, if no REJECT message is received;
otherwise, the actual value of $t_i^{w1}(r)$ depends upon the reception time of REJECT message. In 
this analysis, we have considered $t^{w1}(r)$ as uniform random variable. The total time 
spend by node $i$ in VS state of round $r$ is  denoted by $t_i^{VS}(r)$, and it is equal to 
$t_i^{w}(r) + t^{req} + t_i^{w1}(r)$.
\item Finally, node $i$ enters SS state from VS state, if it receives GRANT messages from all of 
its one hop neighbors for the transmitted REQ message. We denote by $t_i^{SS}$  the total time 
taken by node $i$ to reach SS state, and let $R_i$ be the corresponding number of rounds.
The value of  $t_i^{SS}$  can be written as:

\vspace{-13pt}

\begin{align}
  t_i^{SS} & =  \sum_{r=1}^{R_i} t_i(r)      
\end{align}
\end{itemize}
, where  $t_i(r)$ is defined as the duration of  round $r$.
The above summation is an example of sum of  independent random 
variables ($t_i(r)$) for   total number of rounds ($R_i$), which is also a random variable.
To calculate the expected value of $t_i^{SS}$, we use conditional expectation and the
law of iterated expectation, as follows. \newline

\vspace{-5pt}

\noindent$
\begin{array}{r@{}l}
 & E[t_i^{SS}/(R_i = n)]  = E\left[\sum_{r=1}^{n} t_i(r)\right]  = \sum_{r = 1}^{n} E[t_i(r)] \vspace{3pt} \\ 
 & = \sum_{r = 1}^{n} E[t_i^{CS}(r) + t_i^{VS}(r)] \vspace{3pt}  \\
 & =  \sum_{r = 1}^{n} E[t_i^{CS}(r) + t_i^{w}(r) + t^{req} + t_i^{w1}(r)] \vspace{3pt}     \\ 
 & =  \sum_{r = 1}^{n} E\left[t_i^{CS}\right] +  E[t_i^w(r)] + E[t^{req}] + E[t_i^{w1}(r)] \vspace{3pt} \\
 & = n(t_i^{CS} + t^{req}] + \sum_{r = 1}^{n} \left(  E[t_i^w(r)] + E[t_i^{w1}(r)] \right) \vspace{3pt}  \\ 
 & = n(t^{CS} + t^{req}) + \sum_{r = 1}^{n} \left( S/2 + S/2\right) \vspace{3pt}  \\  
 & = n(t^{CS} + t^{req} + S)    
\end{array}
$


\noindent Therefore, 
  \begin{align}
  E[t_i^{SS}/R_i] &= R_i(t_i^{CS} + t^{req} + S) \nonumber   \\ 
  E[t_i^{SS}] &=  E[E[t_i^{SS}/R_i]] =  E[R_i(t^{CS} + t^{req} + S)] \nonumber   \\ 
  &= E[R_i]*(t^{CS} + t^{req} + S)
  \end{align}

  To calculate the expected value of $R_i$, we define $q_i(r)$ as the probability of node $i$ 
entering  the SS state from VS state in  round $r$. The value of E[$R_i$] can be calculated as 
follows. 
\abovedisplayskip=2pt
\belowdisplayskip=2pt
\begin{align}
   E[R_i] & =  \sum_{r=1}^ \infty r * P(R_i = r) \nonumber \\
   & = \sum_{r=1}^ \infty r * \left(\prod_{j=1}^{r-1}(1 -  q_i(j))\right) * q_i(r)
  \end{align}

  Finding a closed form formula for the above sum is not feasible as the value of $q_i(r)$  
depends upon the probability distribution $P_i$, using which the node $i$ tries to get 
different slots, and also on the probability  distribution $P_j$  of its neighboring nodes. The value of 
$P_i$ is not fixed and it keeps changing in every round depending upon the number of nodes in $N2_i$ 
which  have already entered SS state. Instead of trying to find the accurate value of 
$E[R_i]$, we calculate an upper bound on $E[R_i]$, with the help of following theorem. \vspace{2pt}
\begin{theorem}
 Let $q_i^{min} = min(q_i(r):  1\le r \le \infty$). Then  $E[R_i] \le \frac{1}{q_i^{min}}$ \vspace{3pt}
\begin{proof}

\noindent$
\begin{array}{r@{}l}
 
   E[R_i] & = \sum_{r=1}^ \infty r * P(R_i = r) = \sum_{r=1}^ \infty \sum_{t = 1}^r P(R_i = r)  \vspace{3pt} \\
   
   & = \sum_{t=1}^ \infty \sum_{r = t}^\infty P(R_i = r) = \sum_{t=1}^ \infty P(R_i \ge t)  \vspace{3pt} \\
   & = \sum_{t=1}^ \infty  \prod_{r=1}^t (1-q_i(r))  \vspace{3pt}  \\ 
   & \le  \sum_{t=1}^ \infty  \prod_{r=1}^t (1-q_i^{min})  \vspace{3pt} \\
   & = \sum_{t=1}^ \infty (1-q_i^{min})^t =   \frac{1}{q_i^{min}} 
\end{array}
$
\end{proof}
\end{theorem}
 
Node $i$ can enter SS state in a round if 
no node in $N2_i$ is already in VS state or SS state with respect to the same slot, which has 
been selected by node $i$ in VS state in the same round.  Let  $q_{i,s}(r)$ be the  
probability of node $i$ entering  SS state from VS state in round $r$ with slot $s$.  Clearly, the value of $q_{i,s}(r)$ 
depends upon the 
probability with  which node $i$  tries to select slot $s$ and also on the 
probabilities of other  nodes in $N2_i$ with  which they try to select the same slot. The 
equation in (5) gives the value of  $q_{i,s}(r)$ in terms of current slot selection probabilities of  slot 
$s$ at node $i$ and at all the nodes in $N2_i$. Here, we assume that all the nodes 
have  exactly $S-1$ two-hop neighbors, i.e.,  $\forall i: |N2_i| = S-1$.
Note that, $S$ is always taken to be greater than the maximum number of two-hop neighbors for any node in the network. 
Therefore, the above assumption would give the worst case analysis for the 
expected runtime of RD-TDMA algorithm.
\begin{align}
q_{i,s}(r) =  \left(\displaystyle\prod_{j=1, j\ne i}^{S} (1- p_{j,s}(r))\right)* p_{i,s}(r)
\end{align}
, where $p_{j,s}(r)$ is the probability of node $j$ selecting node s in CS state of round $r$. 

%


In order to simplify the Eq. 5, let us rearrange the IDs of the nodes in the following manner. 


\textbullet \textcolor{white}{~} The ID of node i is changed to 1. 

\textbullet \textcolor{white}{~}  The IDs of nodes in $N2_i$ would be set from 2 to $S$. The 
ordering among these nodes could be arbitrary. 

\textbullet  \textcolor{white}{~}  The IDs of all other nodes in the network would become $S+1$ to $n$, where $n$ is the total 
number of nodes in the network. The ordering 
among these nodes could also be arbitrary. 

Note that the above rearrangement will not change the probability of node i entering SS state in a 
round. After this transformation $q_i(r)$, $q_{i,s}(r)$, $q_i^{min}$, and $E[R_i]$ would become  
 $q_1(r)$, $q_{1,s}(r)$, $q_1^{min}$, and $E[R_1]$ respectively. Further, we can also omit the subscript 1 from  
all the expressions for sake of clarity. The Eq. 5  can now be rewritten as, 

\vspace{-11pt}

\begin{align}
q_s(r) =  \left(\prod_{j=2}^{S} (1- p_{j,s}(r))\right)* p_{1,s}(r)
\end{align}

In terms of $q_s(r)$, the value of $q(r)$ can be written as,
\begin{align}
 q(r) = \sum_{s=1}^{S} q_s(r)
\end{align} 

Now our aim is to find out the value of $q^{min}$.
First, we consider a relatively simpler case of single hop WSNs, in which   
all the nodes in the network interfere with each others transmission. 
Let, after $r$ rounds, exactly $m$ number of nodes are in SS state. The remaining $S-m$ 
nodes which are not in SS state, set their slot probability to $\frac{1}{(S-m)}$, for all
unoccupied slots, i.e. $p_{j,s}(r) = \frac{1}{(S-m)}$. Substituting the value of $p_{j,s}(r)$ in Eq. 
6, and further in Eq. 7 we get,

\vspace{-5pt}

\begin{align}
q(r) &=  \sum_{s=1}^{S-m} \left(1 - \frac{1}{S-m}\right)^{S-m-1} * \frac{1}{S-m} \nonumber 
\\
    &=  	\left(1 - \frac{1}{S-m}\right)^{S-m-1}
\end{align}

The value of $q(r)$ as given in Eq. 8,  is a  monotonically decreasing function with respect 
to $S-m$, and converges to $\frac{1}{e}$ as $(S-m) \rightarrow \infty$, since  $\frac{1}{q(r)} = 
\left(1 + \frac{1}{S-m-1}\right)^{S-m-1}$, which is one of the definition of mathematical constant $e$ as 
$S-m-1 \rightarrow \infty$. This concludes that  $q^{min} = 
\frac{1}{e}$. Substituting the value of $E[R_i]$ in Eq. 3 by $\frac{1}{q^{min}}$, we get $E[t_i^{SS}] \le e(t^{CS} + t^{req} + S)$,  and this is of the order of $O(S)$. Note that,  in case of 
single-hop network, the number of slots  $S$ is equal to the number of nodes $n$ in the 
network.
  
Now, we will consider a more generalized situation of multi-hop WSNs, in which not  all the nodes in 
the network interfere with each others transmission. Let us define a binary square matrix, $B$ of 
size $S$, in the following manner.

\vspace{-10pt}

\begin{align}
b_{i,s} =
\left\{
\begin{array}{ll}
 1,  & \mbox{if $p_{i,s}(r) > 0$} \\
0,   & \mbox{otherwise}
\end{array}
\right. \nonumber
\end{align}

The matrix P and B show an example of probability matrix and its corresponding binary 
transformation for $S = 3$. 
\vspace{-10pt}

\begin{align}
P &= \begin{pmatrix}
    1/3 & 1/3 & 1/3 \\
    1/2 & 1/2 & 0 \\
    1/2 & 0 & 1/2 
\end{pmatrix}
&
B &= \begin{pmatrix}
    1 & 1 & 1 \\
    1 & 1 & 0 \\
    1 & 0 & 1 \nonumber
\end{pmatrix}
\end{align}

Let $\beta_{i} = \sum_{s=1}^{S} b_{i,s}$ (number of 1's in row $i$ of matrix B) and 
$\alpha_s = \sum_{i=2}^S b_{i,s}$ (number of 1's in column $s$ of matrix B, excluding first row). The Eq. 5  can be rewritten in terms of $b_{j,s}$ and 
$\beta_{j}, 1 \le j \le S$, as follows.

\vspace{-5pt}

\begin{align}
 q_s(r) = \frac{b_{1,s} }{\beta_{1}}  \prod_{j=2}^{S} \left( 1 - 
\frac{b_{j,s}}{\beta_{j}}\right)
\end{align} 

Let $B^{min}$ be the matrix for which value of $q(r)$ is minimum. To 
find out the properties of $B^{min}$, we start with the hypothesis that 
$q(r)$ would be minimum, if none of the nodes in $N2_1$, is in SS state. This implies that 
node 1 is still transmitting in 
all the slots with probability $\frac{1}{S}$, i.e., $b_{1,s} = 1,  1 \le s \le 
S$.
Now, we will present two lemmas based on the aforementioned hypothesis, and this hypothesis will be used to find 
out the properties of $B^{min}$ in theorem 7, where we also explain the need for it.

\begin{lemma}
For a given instance of matrix B, let  $b_{1,u} = 1, 1 \le u \le S$, and for a 
slot $s$, $q_s(r) \le q_u(r),  \forall u \ne s$.   Then, for any row   i, the value of $q(r)$ reduces or 
remains  
same, if $b_{i,s}$ is changed from 1 to 0.

\begin{proof}
Let $q^{old}(r)$ and $q^{new}(r)$ be the respective value of $q(r)$ before and after the 
conversion of  $b_{i,s}$ from  1 to 0. We need to show that 
$q^{old}(r) \ge q^{new}(r)$. Similarly, $q_s^{old}(r)$ and $q_s^{new}(r)$ can be defined. 
Since $b_{1,u} = 1, \forall u, 1 \le u \le S$,
the $q_s^{old}(r)$ can be written as,

\begin{align}
 q_s^{old}(r) &= \frac{1}{S} \prod_{j=2}^{S} \left(1 - 
\frac{b_{j,s}}{\beta_{j}}\right)  \nonumber \\
 &= \frac{1}{S}\left(\prod_{j=2, j \ne i}^{S} \left(1 - 
\frac{b_{j,s}}{\beta_{j}}\right)\right) * \left(1 - \frac{1}{\beta_{i}}\right) 
\end{align} 
and since  $b_{i,s}$ becomes 0, after the conversion,   $q_s^{new}(r)$ would be,
\begin{align}
 q_s^{new}(r) = \frac{1}{S}\left(\prod_{j=2, j \ne i}^{S} \left(1 - 
\frac{b_{j,s}}{\beta_{j}}\right)\right)
\end{align} 
Therefore, from Eq. (10) and (11), we get,
\abovedisplayskip=2pt
\belowdisplayskip=2pt
\begin{align}
 q_s^{new}(r) - q_s^{old}(r) =   \frac{q_s^{old}(r)}{\beta_{i}-1}  
\end{align} 

Similarly, for all other slots $u \ne s$ and $b_{i,u} = 1$  
\begin{align}
 q_u^{old}(r) &= \frac{1}{S} \prod_{j=2}^{S} \left(1 - 
\frac{b_{j,u}}{\beta_{j}}\right)  \nonumber \\
 &= \frac{1}{S}\left(\prod_{j=2, j \ne i}^{S} \left(1 - 
\frac{b_{j,u}}{\beta_{j}}\right)\right) * \left(1 - \frac{1}{\beta_{i}}\right) 
\end{align} 
and since  $b_{i,s}$ becomes 0, after the conversion, the value of $\beta_i$  will reduced by one. 
\begin{align}
 q_u^{new}(r) = \frac{1}{S}\left(\prod_{j=2, j \ne i}^{S} \left(1 - 
\frac{b_{j,u}}{\beta_{j}}\right)\right) * \left(1 - \frac{1}{\beta_{i}-1}\right) 
\end{align} 
Therefore, from Eq. (13) and (14), we get,
\begin{align}
 q_u^{new}(r) - q_u^{old}(r) =   -\frac{q_u^{old}(r)}{(\beta_{i}-1)^2}  
\end{align}

To show that $q^{old}(r) \ge q^{new}(r)$, we calculate  $q^{new}(r) - q^{old}(r)$ as 
follows,
\begin{align}
& q^{new}(r) - q^{old}(r) = \sum_{u = 1}^S q_u^{new}(r) - \sum_{u = 1}^S 
q_u^{old}(r)  \nonumber \\
 &=\left(\sum_{u \ne s, b_{i,u} = 1}^S q_u^{new}(r) -  q_u^{old}(r)\right)  + \left(q_s^{new}(r) -  
q_s^{old}(r)\right)  \nonumber \\
 &= \left(\sum_{u \ne s, b_{i,u} = 1}^S - \frac{q_u^{old}(r)}{(\beta_{i}-1)^2}\right) + 
\frac{q_s^{old}(r)}{\beta_{i}-1}   \nonumber \\
 &= \left(\sum_{u \ne s, b_{i,u} = 1}^S  \frac{-q_s^{old}(r) + 
(q_s^{old}(r)-q_u^{old}(r))}{(\beta_{i}-1)^2}\right) +  \frac{q_s^{old}(r)}{\beta_{i}-1}  \nonumber \\
&= \left(\sum_{u \ne s, b_{i,u} = 1}^S \frac{-q_s^{old}(r)}{(\beta_{i}-1)^2}\right) \nonumber \\ 
& ~~~~~~~  + \left(\sum_{u 
\ne s, b_{i,u} = 1}^S \frac{(q_s^{old}(r)-q_u^{old}(r))}{(\beta_{i}-1)^2}\right) 	+ \frac{q_s^{old}(r)}{\beta_{i}-1}   \nonumber
\end{align}

Since the number of 1's in row i is $\beta_{i}$, the number of terms in the first summation of  
above equation would be exactly  $\beta_{i} - 1$. 
Therefore,  
\begin{align}
 q^{new}(r) - q^{old}(r) 	= \sum_{u \ne s, b_{i,u} = 1}^S \frac{q_s^{old}(r)-q_u^{old}(r)}{(\beta_{i}-1)^2} \le 0 \nonumber \\
    \because q_s(r) \le q_u(r), \forall u \ne s \nonumber
\end{align}
\end{proof}
\end{lemma}

\begin{lemma}
  For a given instance of matrix B, let $b_{1,s} = 1, \forall s, 1 \le s \le S$ and  $\beta_{i} = 2, \forall i, 2 \le i \le 
S$.
  Then the following holds. For any two columns s and u and, for any row i, such that $\alpha_{s} > \alpha_{u}$,  $b_{i,s} = 1$  and 
$b_{i,u} = 0$, 
   $q(r)$ either reduces or remains the same if the value of $b_{i,s}$  and $b_{i,u}$ are interchanged.  
\end{lemma}
\begin{proof}
Consider $q^{old}(r)$, $q^{new}(r)$, $q_s^{old}(r)$ and $q_s^{new}(r)$ as defined in lemma 2. We need to show that 
$q^{old}(r) \ge q^{new}(r)$. Here, $q_s^{old}(r) = \frac{1}{S} * 1/2^{\alpha_s}, q_u^{old}(r) = \frac{1}{S}  * 
1/2^{\alpha_u}, q_s^{new}(r) = \frac{1}{S}* 1/2^{\alpha_{s}-1}, q_u^{new}(r) = \frac{1}{S}*  1/2^{\alpha_{u}+1}$. 
Since columns other than $u$ and $s$ remain the same, we have
\begin{align}
 &q^{new}(r) - q^{old}(r) = (q_s^{new}(r) + q_u^{new}(r)) - (q_s^{old}(r) + q_u^{old}(r)) \nonumber \\
 &=  (q_s^{new}(r) - q_s^{old}(r)) + (q_u^{new}(r) - q_u^{old}(r) \nonumber \\
 &= \left(\frac{1}{2^{\alpha_s-1}} -  \frac{1}{2^{\alpha_s}}\right) + \left(\frac{1}{2^{\alpha_u+1}} -  \frac{1}{2^{\alpha_u}}\right)  
\nonumber \\
 & =  \frac{1}{2^{\alpha_s}} - \frac{1}{2^{\alpha_u+1}}  \le 0 \;\;\;\;\;\;\;\;\;\; \because \alpha_{s} > \alpha_{u} \nonumber  
 \end{align}
\end{proof}

Now we  prove that  $B^{min}$ would satisfy a few constraints,
in terms of $\alpha_u$ and $\beta_i$
, with the help of lemma 3 and 4.

\begin{theorem}
$B^{min}$ has the following properties.
\begin{enumerate}
 \item $\beta_{i} = 2, ~ \forall i \text{  in the range  } [2, S] $
 \item For exactly two columns s1 and s2, $\alpha_{s1} = \alpha_{s2} = 1$ and for all other columns $u \ne s1,s2$, $\alpha_{u} = 2$.
\end{enumerate}

\begin{proof}
We prove both the properties for two different cases:  $\beta_1 = S$ and $\beta_1 \ne S$

\noindent{\textbf{Case 1. $\beta_1 = S$:}} 
The property (1) can be proved by contradiction. First, we show that \textbf{$\beta_{i} \ge 2$, for $2 \le i \le S$}. If  $\beta_i = 1$
with $b_{i,s} = 1$, for some row i, then node i is in SS state. Therefore, node 1 should have stopped transmitting in slot $s$, i.e.  
$b_{1,s} = 0$, which contradicts our assumption that $\beta_1 = S$. Now, we show that  $\beta_{i} = 2$, for $2 \le i \le S$.
Let $\exists i: \beta_{i} > 2$ and $\mathcal{A}$ be the set of 
column indexes u for which $b_{i,u} = 1$, then $\exists s \in$ $\mathcal{A}$,  such that $q_s(r) \le q_u(r), \forall u \in A$.
Therefore, by the virtue of lemma 1, $q(r)$ reduces or remains same, if $b_{i,s}$ is changed from 1 to 0. 
The same process can be repeated till $\beta_{i} = 2$. 
 
The property (2) can also be proved by contradiction. We know that $\beta_{1} = S$ and  $\beta_{i} = 2$, for $2 \le i \le S$. 
Therefore, $\sum_{s=1}^{S} \alpha_s = 2(S-1)$.  
First, we show that $\alpha_{s} \le 2$, for $1 \le s \le S$. 
For a column s, $\alpha_{s} > 2 \implies \exists u: \alpha_{u} < 2$, otherwise  $\sum_{s=1}^{S} \alpha_s$ would become less than 
$2(S-1)$.
In this case, for any row i, such that $b_{i,s} = 1$ and $b_{i,u} = 0$, can be interchanged, by virtue of lemma 3. This proves that 
$\alpha_{s}$ could be either  0,1 or 2,  for 1 $\le s\le S$. 
Since, any column can have at most two, 1's, this implies that at most one column of type $\alpha_{s} = 0$ can exist and that also can be changed to a column with $\alpha_s=1$,  by virtue of lemma 3.
Furthermore, the number of columns of type  $\alpha_{s} = 1$ cannot be one, since 
$2(S-1)$ is  even. Finally, we can say that number of columns of type  $\alpha_{s} = 1$ is exactly 2; otherwise,
the total sum will be less than $2(S-1)$.

\noindent \textbf{Case 2. $\beta_1 \ne S$:}
Let $q_{case 1}(r)$ and $q_{case 2}(r)$  be the corresponding summation for case 1 and case2, respectively. 
The value of  $q_{case 1}(r)$ would  be $\frac{S+2}{4S}$ using equation (7) and (9). We will prove that $q_{case 1}(r) < q_{case 2}(r)$ 
by 
showing
that, any perturbation in the  matrix corresponding to case 1, will increase the value of $q(r)$.
We have already proved, in case 1, that any modification in any of the row from 2 to row $S$ and leaving row 1 unchanged, will 
increase 
$q(r)$.
Now, let us change a single entry 
$b_{1,u} =$ 1 to 0, i.e., node 1 has decided not to transmit in slot u. This only happens when at least one adjacent  node i 
in $N2_1$ has already entered  SS state for slot u,  which implies that $b_{i,u} = 1$ and $b_{i,s} = 0, \forall s \ne u$. 
Let us interchange the row i with row $S$ and column $u$ with column $S$. In this case,  
$b_{1,S} =  0$, $b_{1,s} = 1$, $\forall s\ne S$,
$b_{S,S} =  1$
and $b_{S,s} = 0, \forall s\ne S$. Consider the sub matrix of size $S-1$ times $S-1$. 
The minimum value of $q(r)$ which can be achieved by this sub matrix would be $\frac{S+1}{4(S-1)}$. 
Moreover, $q_s(r) = 0$, because $b_{1,S} =  0$. Therefore, $q_{case2}(r)$ = $\frac{S+1}{4(S-1)}
> \frac{S+2}{4S} = q_{case1}(r)$.
\end{proof}
\end{theorem}

From Eq. 7, we know that,  $q^{min}$ can be  achieved when $q(r)$ is minimum and $B^{min}$ should satisfy the properties as given 
in theorem 7. 
Therefore,

\begin{align}
 q^{min} = 1 - \left(\frac{4*S-1}{4*S}\right)^{(S-2)}*\left(\frac{2*S-1}{2*S}\right)^{2}
\end{align}

%
%

\begin{theorem}
   The value of $q^{min}$ converges to  $1 - e^{\frac{-1}{4}} = 0.221$, as $S \to \infty$. 
\begin{proof}
 \begin{align}
   &\displaystyle\lim_{S \to \infty}q^{min} = \nonumber \\
   &=   \displaystyle\lim_{S \to \infty} \left(1 - 
\left(\frac{4*S-1}{4*S}\right)^{(S-2)}*\left(\frac{2*S-1}{2*S}\right)^{2}\right)	\nonumber \\
&=  \displaystyle\lim_{S \to \infty} \left(1 - 
\left(1 +  \frac{-1}{4*S}\right)^{S}\right) =   1 - e^{\frac{-1}{4}} = 0.221  \nonumber 
\end{align}
\end{proof}
\end{theorem}

Substituting $E[R_i]$ with $\frac{1}{q^{min}}$ in Eq. 3, from theorem 6  we get

\abovedisplayskip=2pt
\belowdisplayskip=2pt
\begin{align}
 E[t_i^{SS}] \le  \frac{(t^{CS} + t^{req} + S)}{\left( 1 - 
\left(\frac{4*S-1}{4*S}\right)^{(S-2)}*\left(\frac{2*S-1}{2*S}\right)^{2} \right)}
\end{align}

Finally,  the value of E[T], i.e., the expected time taken by RD-TDMA algorithm to perform scheduling can be calculated as
\abovedisplayskip=2pt
\belowdisplayskip=2pt
\begin{align}
 E[T] &= E\left[max\left(t_1^{SS}, t_2^{SS} \dots t_n^{SS}\right)\right] \nonumber \\
 &=  (t^{CS} + t^{req} + S)* E\left[max\left(R_1, R_2 \dots R_n\right)\right] 
\end{align}
The $R_i$s can be assumed as i.i.d (independent and identically distributed) geometric random variable with parameter $q^{min}$. In this case, the $E[t_1^{SS}]$ would be higher than the actual expected time to enter SS state, by node i. 
Let $R = max\left(R_1, R_2 \dots R_n\right)$. The value of $E[R]$ can be calculated as, 

\vspace{-5pt}

\begin{center}
\noindent$
\begin{array}{r@{}l}

    E[R] &= \sum_{r \ge 0} P\left([R > r]\right)   \vspace{3pt} = \sum_{r\ge0} (1 - P(R \le r)) \vspace{3pt} \\
   & = \sum_{r\ge0} (1 - P(R_i \le r)^n)  \vspace{3pt}   =  \sum_{r\ge0} \left(1 - (1 - \mu^r)^n\right) 
   \end{array}
$
\end{center}

, where $\mu = 1 - q^{min}$. By considering the above infinite sum as right and left hand Riemann sum approximations of the corresponding integral, we obtain, 
\abovedisplayskip=2pt
\belowdisplayskip=2pt
\begin{align}
 \int_0^\infty  \left(1 - (1 - \mu^r)^n\right) du \le E[R] \le 1 +  \int_0^\infty \left(1 - (1 - \mu^r)^n\right) du
\end{align}

With the change of variable $w = 1 - \mu^r$,  we have,
\begin{align}
  E[R] &\le 1 +\frac{1}{log \: \mu}  \int_0^1 \frac{1-w^n}{1-w} dw \nonumber \\
 &= 1 + \frac{1}{log \:\mu }  \int_0^1 (1 + w + ... + w^{n-1}) dw \nonumber \\
 &= 1 + \frac{1}{log \:\mu } (1 + \frac{1}{2} + ... + \frac{1}{n})  \approx 1 - \frac{log \: n}{log \: \mu} 
\end{align}

We know from Eq. 16, that the value of $q^{min}$ depends only upon $S$, 
which is also a measure of neighborhood density of  multihop WSNs, therefore for constant $n$,  
the $E[R]$ (expectation of maximum number of rounds taken by any node in the network to reach SS state) is of the order of $O(log \: S)$.  A more 
rigorous analysis on expectation of the maximum of IID geometric random variables can be found in \cite{MAX_IID}. 
After substituting the value of E[R] from Eq 20 in Eq. 18,  we get $E[T] =  \left(t^{CS} + t^{req} + S\right) * \left(1 - \frac{log \: 
n}{log \: \mu}\right)$, and it is of the order of   $O(S \: log \: S)$, assuming $n$, $t^{CS}$ and $t^{req}$ as constants. 


%
%

\subsection{Runtime of DSLR Algorithm}
Let $X_{n}$ be the number of moves made by a node to reach its final  slot, during the execution of DSLR algorithm,  
where $n$ is the initial slot Id of the node. Clearly $X_{n}$ is a random variable which can take the values ranging from $1$ to $n$, and its value 
depends upon the initial TDMA-Schedule provided to the DSLR algorithm, sender-receiver relationship and network topology. We assume that the 
probability of a slot being free for a node is uniformly distributed. We model the behavior of a particular node using a 
discrete time markov chain (DTMC),
with the current slot Id of the node as  state of  the markov chain. The transition probabilities, $\pi_{i,j}$, are defined as follows.

\vspace{-10pt}

\[ \pi_{i,j} =
\left\{
\begin{array}{ll}
\frac{1}{i-1}, & \mbox{j $<$ i} \\
0,   & \mbox{otherwise}
\end{array}
\right. 
\]

, where $i$ is the current slot Id and $j$ is the slot Id to which the transition will ocuur.   
We can derive equations for expected number of rounds required to reach absorbing state (final slot) starting from state $n$, i.e.,  
$E[X_{n}]$, by using the total expectation theorem. Time  to reach absorbing state  starting from  transient state
$n$ is equal to 1 plus the expected time to reach absorption state starting from the next transient state $l$ with probability $\pi_{n,l}$. 
This leads to a system of linear equations which is stated below.

\vspace{-10pt}

\begin{align}
 E[X_{n}] =  0, \text{~~~~~~~~~~~~~~~~~~~~~~~~~~~~~~~~~for absorbing state} \nonumber
\end{align}

\vspace{-10pt}

\begin{align}
 E[X_{n}] =  1 +  \frac{1}{n-1} \displaystyle \sum_{l=1}^{n-1} E[X_{l}], \text{~~~~for non-absorbing states}  \nonumber 
\end{align}

The starting state and the absorbing state for different nodes can be different. The expected time to reach the absorbing state would be 
maximum if starting state is $S$ (initial slot Id of the node is $S$) and the absorbing state is $1$ (final slot Id of the node is $1$). 
 To get the upper bound on expected time, we  solve the above equation for  $E[X_{S}]$ considering absorbing state as $1$, and get,

\vspace{-10pt}

\begin{align}
 E[X_{S}] =  \displaystyle \sum_{l=1}^{n-1} \frac{1}{l} \nonumber 
\end{align}

The above equation is a  harmonic series which can be approximated as $\log n$, for large values of $n$.
Therefore, 

\vspace{-5pt}

\begin{align}
 E[X_{S}] \approx \log n
\end{align}

Let $R$ be the number of rounds spent by a node in a particular slot before moving to its first-free slot.
The value of $R$ depends upon following two properties. 

\vspace{5pt}

\noindent 1) Due to loss of HELLO messages in different frames of a round, a node $i$ cannot move to its first-free slot if 
 $FF_i < slot_i$. In this case, nodes have to wait till they receive all the  HELLO messages  transmitted by their neighbors,
 in a round.
 
Let $p$ be	 the probability that node i will not not 
receive HELLO message from a node in round $R_i$. Then, the probability, P, that node $i$ will receive HELLO messages 
transmitted by all the neighboring nodes in a round, to get the
exact status of slots, can be given as,  
 $P = (1 - p)^{S*4}$.
Finally, the expected number of rounds  that a node has to wait  in a particular slot before moving to its first-free slot, due to loss of 
HELLO message, is $\frac{1}{P}$. 

\vspace{5pt}

\noindent 2) A slot may become free for a node $i$, only after the movement of some other node in its neighborhood. This can be seen as the ``ripple 
effect'', where movement of some node at distant $k$ hop from node $i$ may create a free-slot for the node at distant $k-1$ hop. This 
process continues and eventually  creates a free-slot at node $i$.


By combining the effect of properties (1) and (2), as given above, the upper bound on the expectation of  R  can be calculated as,

\vspace{-13pt}

\begin{align}
E[R] = \frac{1}{P} * D 
\end{align}

\noindent where D is the maximum distance between any two nodes in the graph, in terms of number of hops. 

Finally, the upper bound on the expected runtime of DSLR algorithm, i.e., the  expected number of rounds taken by a node to reach its 
final slot can be calculate by multiplying expected number of moves (Eq. 21) with the expected number of rounds  spent in a slot before moving to 
another slot (Eq. 22). 

\vspace{-13pt}

\begin{align}
T_{DSLR} = \frac{\log n * D}{ (1 - p)^{S*4}} 
\end{align}

\section{SIMULATION RESULTS}
  We have used Castalia  Simulator \cite{WSN_S_0001} to study the performance of proposed algorithm 
in terms of   time and message complexity.   The lognormal shadowing channel model is used to get the accurate estimates for average 
path loss at different receivers. The various values of path loss exponent and  Gaussian zero-mean random-variable are used  to 
experiment with different packet error rate (PER).
    The simulations have been performed for different data rates to see the effect of data 
transmission rate on various performance metrics.
   All nodes are distributed randomly within 250mX250m area. 
  The neighborhood size of the network is changed by varying the number of nodes from 50 to 300. 
  This setup produces topologies with different neighborhood density, varying between 5 and 50.

Figure \ref{fig_MH_latency} shows the runtime of RD-TDMA algorithm with 
respect to $S$ (a measure of two-hop network density), for different data rates.
The runtime increases with respect to the size of the network. However,  runtime increases 
rapidly for larger network density, because of increase in the  
number of message exchanges, which, in turn, leads to higher rate of collision. 
Furthermore, the higher data rates  achieve
less runtime for a particular  network density. This is due to fact that fixed size message can 
be transmitted in lesser time 
with higher data rates. For higher data rates, the number of collisions are less,
and therefore, the runtime is further reduced.

Figure \ref{fig_MH_overhead} shows the average number of messages transmitted by a node during the 
execution of RD-TDMA, to get scheduled, with respect to $S$, for different data rates. 
The message overhead increases  with respect to the size of the network. 
 Similar to the runtime, the increase in message overhead is also not linear. 
The variation in data rate 
should only affect the transmission time of a message and not the number of messages transmitted by 
node. But, the message overhead is less for higher data rates due to lesser number of collisions.   
The important point to be noted here is that
the runtime and message overhead only depends on two-hop network density, instead of on the number of nodes in 
the network.

\begin{figure}[t]
\begin{center}
\includegraphics[scale=.6]{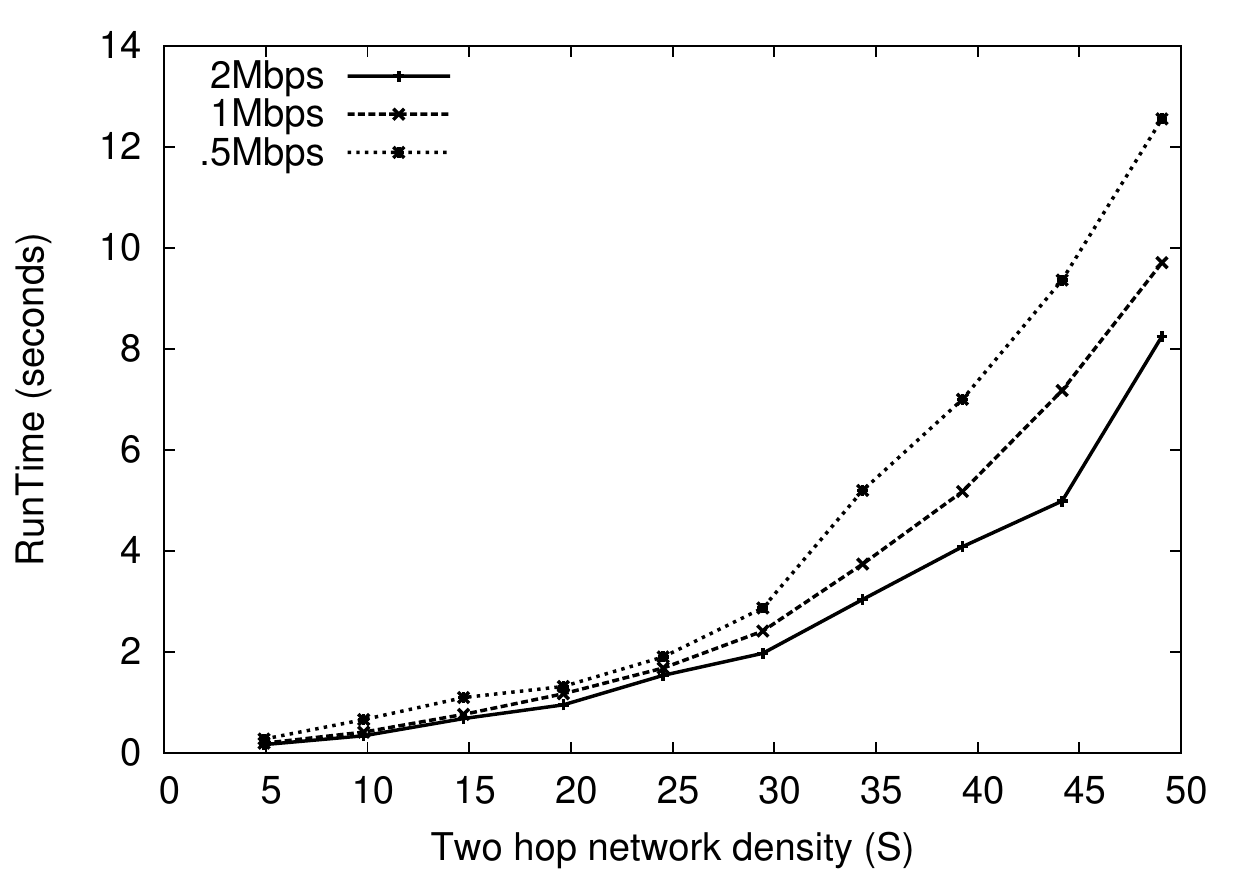}
\label{fig_MH_latency}
\end{center}
\caption{Runtime performance of RD-TDMA for a  multi-hop WSN}
 \end{figure}

\begin{figure}[t]
\begin{center}
\includegraphics[scale=.6]{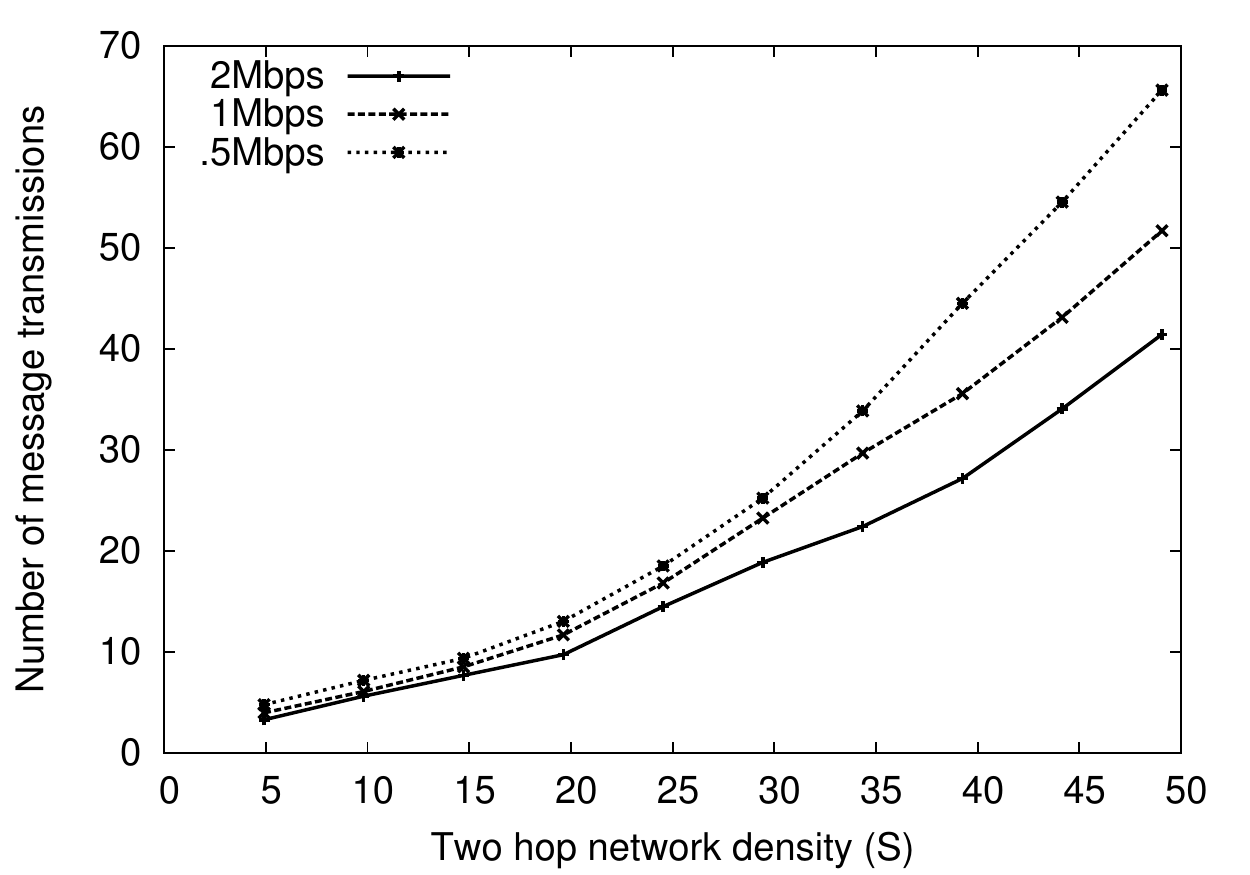}
\caption{Message overhead of RD-TDMA for a  multi-hop WSN}
\label{fig_MH_overhead}
\end{center}
 \end{figure}

\begin{figure}[t]
\begin{center}
\includegraphics[scale=.6]{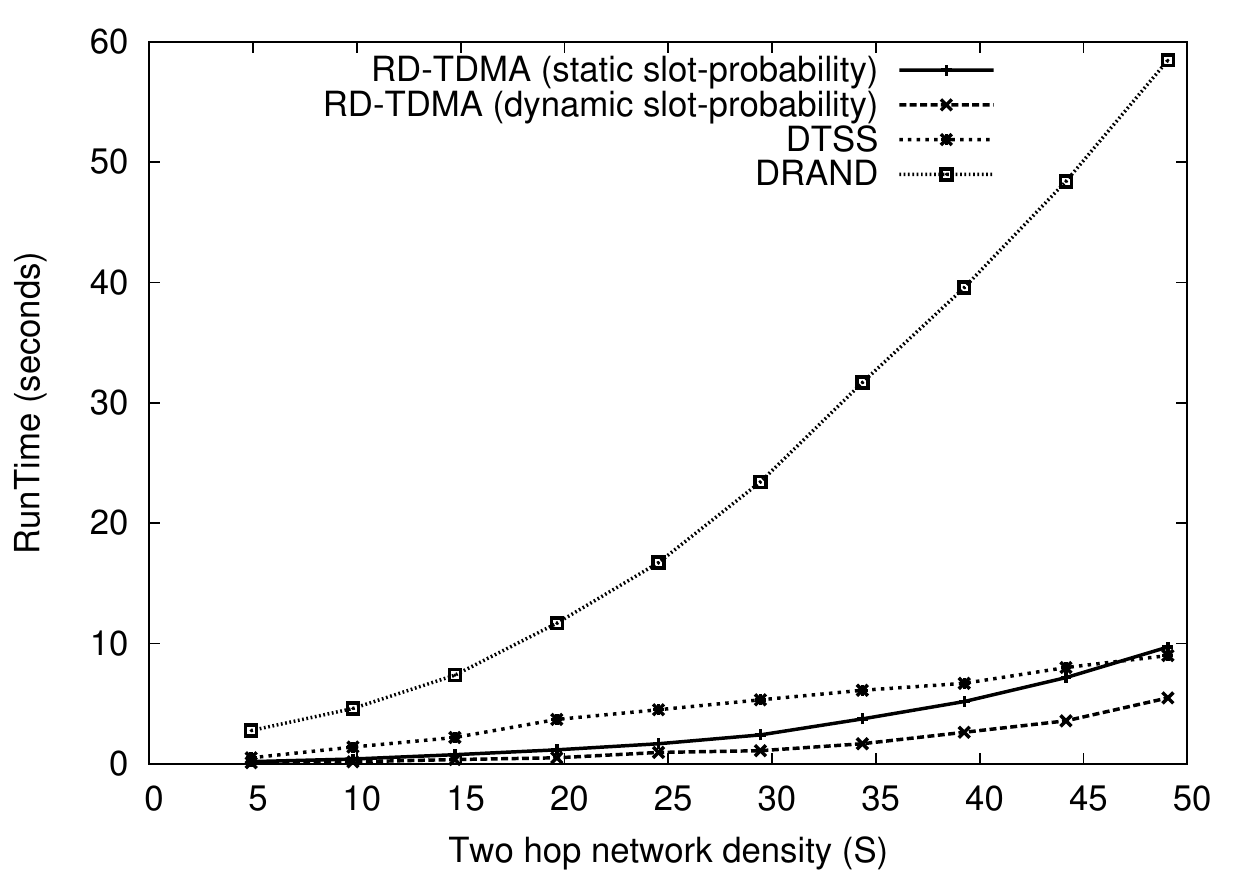}
\caption{Performance improvement with dynamic-slot-probability}
\label{fig_MH_latency_comparision}
\end{center}
 \end{figure}

%

We now analyze the performance improvement due to updating of  slot-probability dynamically,  and also 
compare the performance of RD-TDMA algorithm with that 
of DRAND \cite{DRAND} and 
DTSS \cite{DTSS}
algorithms. 
Figure \ref{fig_MH_latency_comparision} shows the runtime performance of RD-TDMA algorithm for 
static and dynamic slot-probability  along 
with DRAND and DTSS, for a multihop network. The primary reason of getting less runtime in RD-TDMA 
and DTSS is because  both algorithms use probabilistic a approach to  
generate a feasible schedule, which is not necessarily optimal, whereas DRAND algorithm tries to generate a sub-optimal feasible  schedule 
by using greedy approach, which is inherently sequential. 
Finally, we can see that dynamic update of  slot probability, reduces 
runtime to almost 50\% as compared to static probability assignment.

%
%
%

\begin{figure}[t]
\begin{center}
\includegraphics[scale=.6]{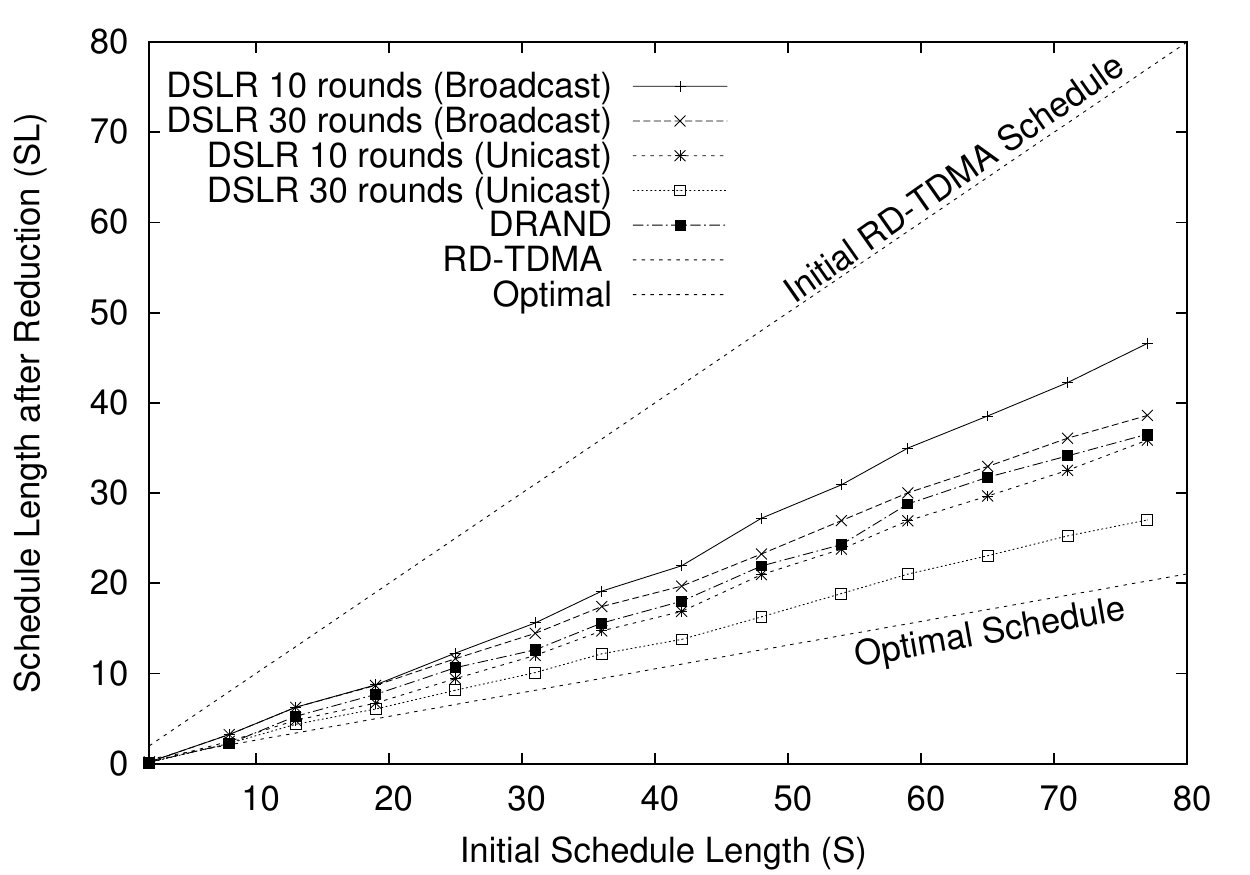}
 \caption{Performance of DSLR algorithm in terms of generated schedule length, with respect to initial schedule, for different
rounds of execution. }
 \label{fig2}
\end{center}
 \end{figure}

\begin{figure}[t]
\begin{center}
\includegraphics[scale=.6]{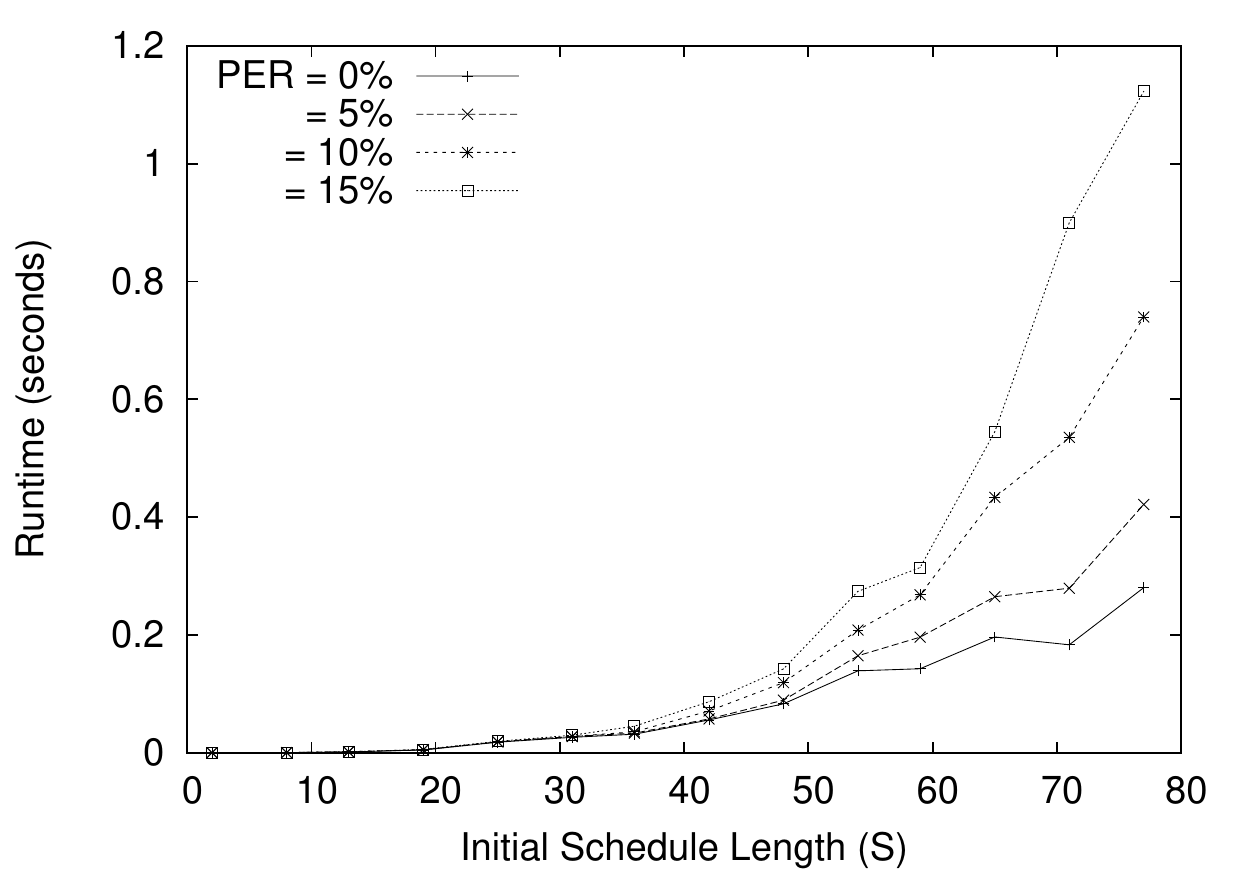}
 \caption{Performance of DSLR algorithm in terms of runtime with respect to initial schedule length, for different PER values.}
 \label{fig3}
\end{center}
 \end{figure}

\begin{figure}[!h]
\begin{center}
\includegraphics[scale=.6]{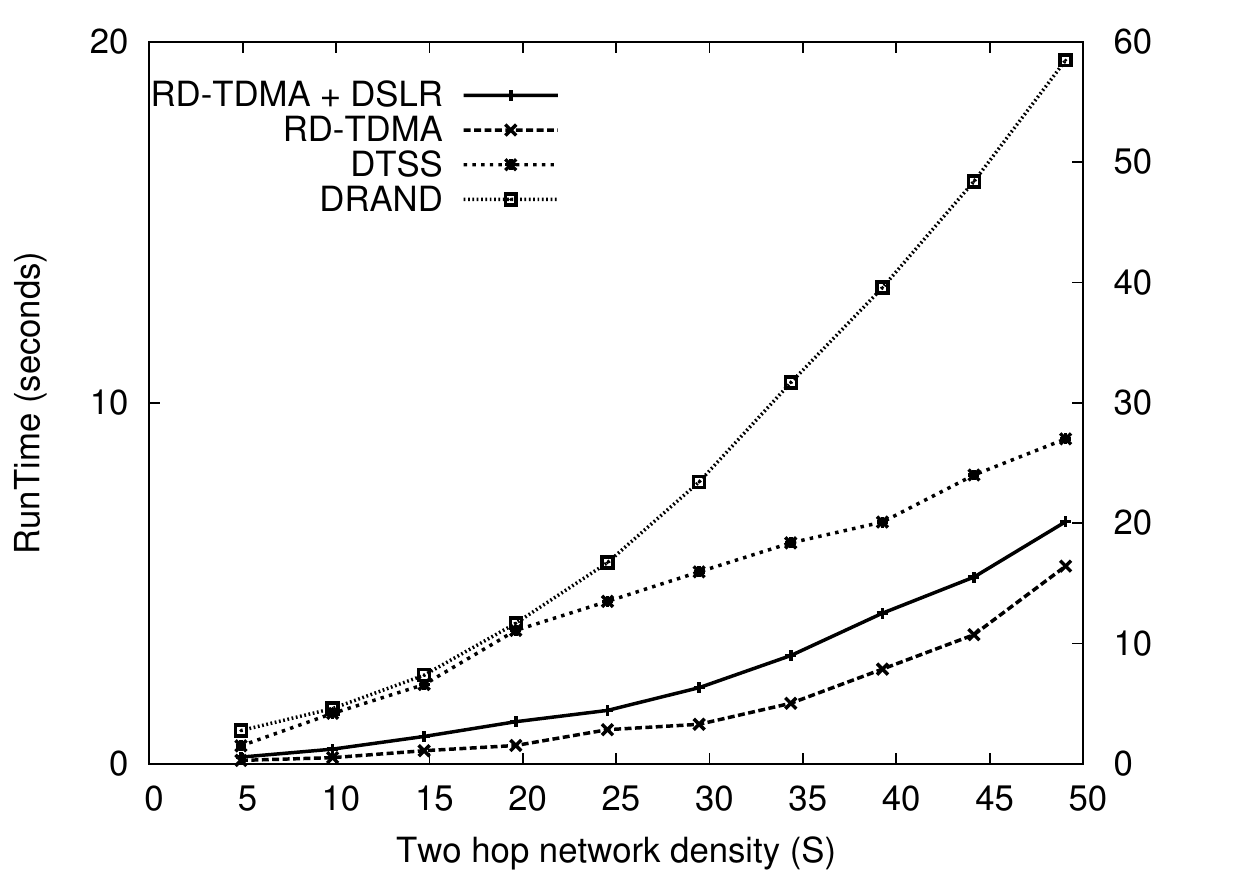}
 \caption{Comparison of DSLR algorithm with existing  TDMA scheduling algorithms. The right side scale $(0-60)$ is for DRAND algorithm.}
 \label{fig4}
\end{center}
 \end{figure}

\begin{figure}[!h]
 \centering
 \includegraphics[scale=.6]{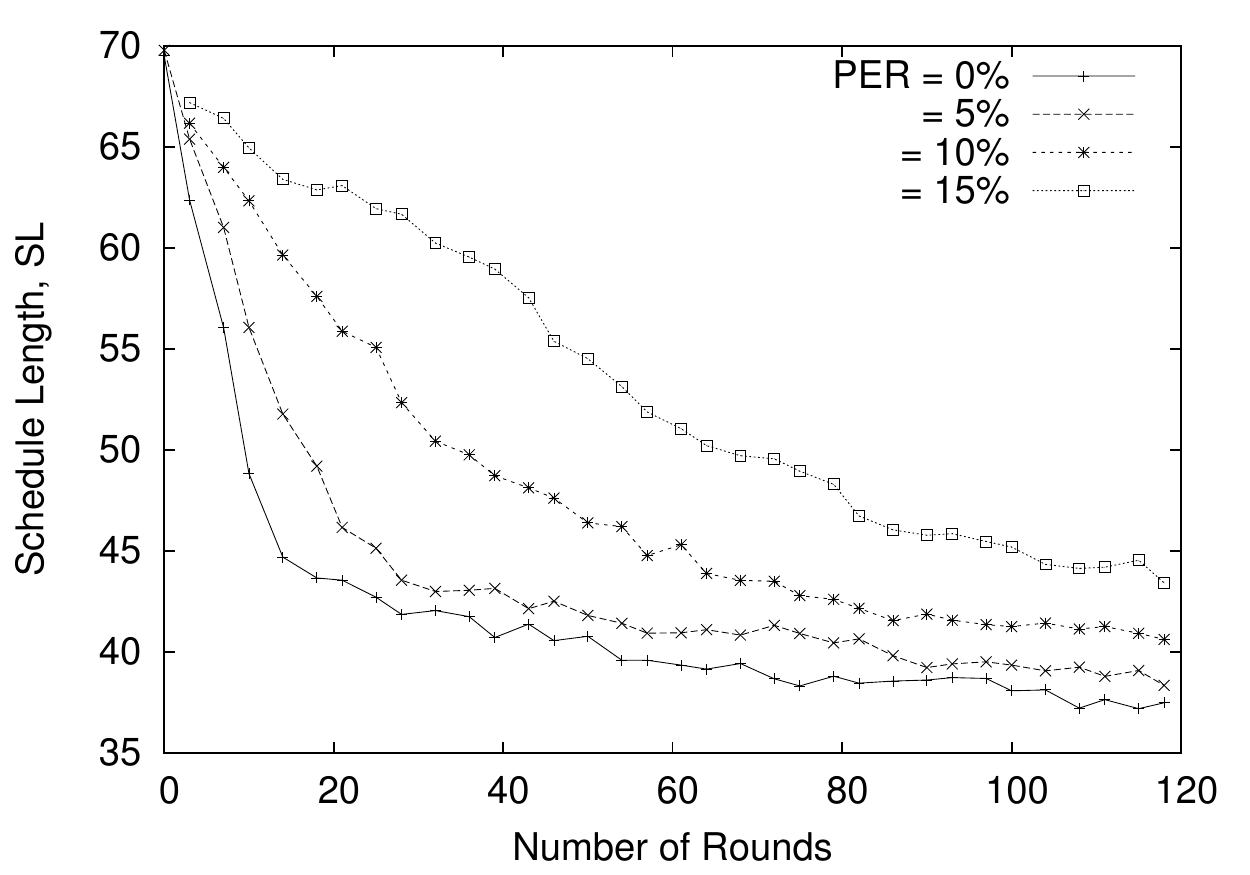}
 \caption{Performance of DSLR algorithm to trade-off runtime
(number of rounds) with schedule length in the presence of
transmission failures.
}
 \label{fig5}
\end{figure}

%
%
%

We now analyze the performance DSLR algorithm in terms of  schedule length achieved after compaction,
time required to perform the compaction, and its capability to trade-off the runtime with generated schedule length.

\noindent \textbf{Schedule length after compaction:} Fig. \ref{fig2} shows the degree
of compaction the DSLR algorithm is able to achieve in 
terms of  schedule length $SL$ after reduction,  with respect to initial schedule length $S$, for broadcast as well unicast
mode of communications, after running the algorithm for
different number of rounds. Fig.  \ref{fig2} also shows the performance
comparison of DSLR algorithm with DRAND algorithm \cite{DRAND},
which uses the greedy graph colouring approach to perform the
TDMA-scheduling. The graph labeled ''Optimal`` represents
the lower bound on schedule length that any optimal algorithm
based on two-hop interference model can achieve. 

Fig. \ref{fig2} shows almost a linear relationship between $F$
and $SL$. Note that, the schedule length produced by DSLR
algorithm, for a particular topology is upper bounded by $S$.
But, in practice, the generated schedule length is far less than
$S$. The DSLR algorithm is able to match its performance with
DRAND algorithm  for broadcast scheduling, and for unicast scheduling
its performance is better than DRAND.

%

\noindent \textbf{Runtime:} 
Fig. \ref{fig3} shows the runtime of the DSLR algorithm to achieve
50\% compaction, i.e., schedule length = $S/2$ with respect to $S$, 
for different PER (Packet Error Rate) values.
We can see that the runtime increases rapidly
with the increase in $S$. This is due to the fact that not only
the number of slots, but the size of each slot also increases (due to increase in the size of HELLO messages)
with respect to increase in $S$, and therefore the size of the
frame (in terms of time) is a quadratic function of $S$. 
Similarly,
the runtime increase rapidly as the PER value increases.
Note that we have considered PER value corresponding to
the packet size of 1k bits. The actual PER value for a particular
experiment depends upon the size of HELLO message, which
increases linearly with the increase in $S$. The probability that a
node will not receive at least one HELLO message in a round
is more for higher values of $S$.


Fig. \ref{fig4} shows the runtime performance comparison of DSLR
algorithm with RD-TDMA, DRAND and DTSS algorithms,
with respect to  two-hop network density. Since we
have used RD-TDMA to generate a feasible schedule before
actually starting the slot compaction, the runtime performance
of proposed scheme is shown after adding the runtime of RD-TDMA algorithm.

%

It is to be noted that the runtime performance of
DSLR algorithm is less than that of DTSS algorithm for which
the generated schedule length is very large. On the other hand,
the schedule length achieved after compaction is approximately
equal to DRAND algorithm, which takes very large time $(\approx
60s)$ to perform the scheduling.

\noindent \textbf{Tradeoff between schedule length and runtime:} Now we
discuss a unique feature of the DSLR algorithm, which is its
capability to to trade-off schedule length with runtime performance as 
per the current channel conditions and application
requirements. Fig. \ref{fig5} shows the schedule length achieved by the
DSLR algorithm with respect to the execution of the algorithm
for increasing number of rounds. We can see that considerable
amount of reduction in schedule length can be achieved in
40 rounds. In case of low duty cycle application, where the primary concern is to save the sensor nodes
energy by enabling them to sleep as long as possible, instead
of network bandwidth, we can select the algorithm to run for a
smaller number of rounds. On the other hand if the bandwidth
demand of the application is high and the same schedule is to  be
used over a long of period time, then we can run the algorithm
for a larger number of rounds to get a more compact schedule.

\section{Conclusions}
The proposed two-phase scheme for TDMA scheduling provides a better schedule length and runtime performance 
than the existing TDMA-scheduling algorithms which are either static or dynamic in nature. 
The proposed RD-TDMA algorithm takes very less time (less than 4s upto two-hop network density, 30) to perform the 
scheduling as compared to other existing distributed scheduling algorithms. 
With the use of dynamic slot-probability updation, we are able to further reduce the runtime by 
50\%. The proposed scheme for distributed schedule length reduction  improves the
bandwidth utilization in three ways. First, it reduces the length
of an input schedule, if it is already not compact. Second,
it converts the broadcast schedule  to  unicast/multicast schedule  assuming that the sender 
and receiver relationship between the sensor nodes, is already 
known. Hence, even an optimal schedule based on broadcast
scheduling can be further compressed by the DSLR algorithm. 
Third, the frame size need not be the same network wide. The 
nodes can select the frame size based on the maximum slot 
occupied by any other node within its two-hop neighborhood. 
The DSLR algorithm for schedule length reduction has 
the capability to trade-off the generated schedule length with 
the time required to generate the schedule. This feature is 
useful for the developers of WSNs to tune the performance as 
per the requirements of WSN applications.





\end{document}